\documentclass[runningheads]{llncs}
\usepackage[T1]{fontenc}
\usepackage{graphicx}
\usepackage{amssymb,xspace}
\usepackage{todonotes}
\usepackage{microtype}
\usepackage{stmaryrd}
\usepackage{xcolor}
\usepackage{thmtools}
\usepackage{amsmath}
\usepackage{hyperref}
\usepackage[capitalise]{cleveref}
\crefname{observation}{Observation}{Observations}
\Crefname{observation}{Observation}{Observations}
\usepackage{subcaption}
\usepackage{tikz}
\usetikzlibrary{automata,arrows.meta,tikzmark,calc,backgrounds,decorations.markings}
\newcommand{\tup}[1]{\langle #1 \rangle}
\newcommand{\dotp}[1]{\langle #1 \rangle}

\renewcommand{\vec}[1]{\boldsymbol{#1}}

\newcommand{\cA}{\mathcal{A}}

\newcommand{\cD}{\mathcal{D}}
\newcommand{\cF}{\mathcal{F}}

\newcommand{\cV}{\mathcal{V}}

\newcommand{\lshape}{\mathbb{L}_W}
\newcommand{\clshape}{\overline{\mathbb{L}}_W}
\newcommand{\reach}[1]{\mathsf{Reach}(#1)}
\newcommand{\boxreach}[1]{\mathsf{BoxReach}(#1)}
\newcommand{\cone}{\mathsf{cone}}
\newcommand{\intCone}{\mathsf{intCone}}
\newcommand{\lattice}{\mathsf{lattice}}
\newcommand{\norm}[1]{\left\lVert#1\right\rVert}

\newcommand{\ST}{\mbox{ s.t. }}

\newcommand{\spos}{\vec{s}_{\!\nearrow}}
\newcommand{\flow}{\vec{f}_{\!\!\!\nearrow}}
\newcommand{\fhigh}{\vec{f}_{\!\!\nnwarrow}}

\newcommand{\vlow}{\vec{\chi}_{\!\!\!\nearrow}}
\newcommand{\vhigh}{\vec{\chi}_{\!\!\nnwarrow}}

\newcommand{\leads}[1]{\xrightarrow{#1}}
\newcommand{\Nleads}[1]{\xrightarrow{#1}_\bbN}
\newcommand{\Zleads}[1]{\xrightarrow{#1}_\bbZ}
\newcommand{\Bleads}[1]{\xrightarrow{#1}_{\Box}}

\newcommand{\bbN}{\mathbb{N}}
\newcommand{\bbR}{\mathbb{R}}
\newcommand{\bbZ}{\mathbb{Z}}
\newcommand{\posquad}{\bbR^2_{\ge 0}}

\newcommand{\Ackermannc}{\texttt{ACKERMANN-Complete}\xspace}

\newcommand{\VAS}{\mbox{\rm VAS}\xspace}
\newcommand{\dVAS}[1]{#1\mbox{\rm-VAS}\xspace}
\newcommand{\dVASS}[1]{#1\mbox{\rm-VASS}\xspace}
\newcommand{\oVASS}{\dVASS{1}}

\newcommand{\tar}{\circledcirc}
\newcommand{\xc}{{\mathsf{x}}}
\newcommand{\yc}{{\mathsf{y}}}

\newcommand{\eff}{\mathsf{eff}}
\newcommand{\effx}{\mathsf{eff}_\xc}
\newcommand{\effy}{\mathsf{eff}_\yc}
\newcommand{\drop}{\mathsf{drop}}
\newcommand{\peak}{\mathsf{peak}}
\newcommand{\oshoot}{\mathsf{over}}
\newcommand{\mshoot}{\mathsf{maxover}}

\newcommand{\LPS}{\mathsf{LPS}(|Q|,\norm{T})}
\newcommand{\BLPS}{B_{\mathsf{LPS}}}

\spnewtheorem{observation}{Observation}{\itshape}{\rmfamily}

\begin{document}

\title{Box-Reachability in Vector Addition Systems}

\author{
  Shaull Almagor\inst{1}\orcidID{0000-0001-9021-1175} \and
  Itay Hasson\inst{1} \and
  Micha\l{} Pilipczuk\inst{2}\orcidID{0000-0001-7891-1988} \and
  Michael Zaslavski\inst{3}
}

\authorrunning{S. Almagor, I. Hasson, Mi. Pilipczuk and M. Zaslavski}

\institute{
  Department of Computer Science, Technion, Israel \\
  \email{\{shaull@technion.ac.il, itay.h@campus.technion.ac.il\}}
  \and
  University of Warsaw \\
  \email{michal.pilipczuk@mimuw.edu.pl}
  \and
  Unaffiliated
}

\maketitle

%\keywords{VAS, Reachability, Boundedness, Box Reachability, Semilinear}

\begin{abstract}
We consider a variant of reachability in Vector Addition Systems (VAS) dubbed \emph{box reachability}, whereby a vector $\vec{v}\in \bbN^d$ is box-reachable from $\vec{0}$ in a VAS $\cV$ if $\cV$ admits a path from $\vec{0}$ to $\vec{v}$ that not only stays in the positive orthant (as in the standard VAS semantics), but also stays below $\vec{v}$, i.e., within the ``box'' whose opposite corners are $\vec{0}$ and $\vec{v}$.

Our main result is that for two-dimensional VAS, the set of box-reachable vertices almost coincides with the standard reachability set: the two sets coincide for all vectors whose coordinates are both above some threshold $W$.
We also study properties of box-reachability, exploring the differences and similarities with standard reachability.

Technically, our main result is proved using powerful machinery from convex geometry.
\end{abstract}

\section{Introduction}
\label{sec:intro}
Vector Addition Systems (VAS) are a well-established formalism for modelling and reasoning about concurrent systems, hardware and software analysis, biological and chemical processes, and many more. 
In particular, VAS and Petri-Nets are equivalent (in a sense), with the latter being a long-studied model~\cite{schmitz2016complexity}.

Formally, a VAS is given by a finite set of vectors $\cV\subseteq \bbZ^d$. A \emph{trace} in $\cV$ is then a sequence of vectors obtained by summing vectors from $\cV$, provided that all coordinates remain non-negative. The latter requirement is dubbed the \emph{VAS Semantics}.
A vector $\vec{v}$ is \emph{reachable} if there is a trace that starts\footnote{Often a starting vector is also given, but the problem has the same ``flavour'' also when starting from $\vec{0}$.} at $\vec{0}$ and ends in $\vec{v}$.

The fundamental problem for VAS is \emph{reachability}: given a vector $\vec{v}$, decide whether it is reachable in $\cV$. This problem has a rich history, and its exact complexity has been only recently settled as $\Ackermannc$\cite{czerwinski2022reachability,leroux2022reachability}, with better bounds known for dimensions $1$ and $2$.
A more involved problem is characterizing the set of reachable vectors. In general, there is no simple characterization (which is understandable, given the high complexity of reachability). However, in dimension up to five this set is \emph{semilinear}~\cite{hopcroft1979reachability}.
%,blondin2015,leroux2004flatness}.

From a modelling perspective, the coordinates of the vectors in a VAS typically correspond to some resources (e.g., memory locks, battery level, queue size, molecule concentration, etc.). Then, reachability amounts to the question of whether a certain configuration of the resources can be reached. 
In the standard reachability question, no attention is paid to \emph{how} the target configuration is reached. In particular, it may be the case that in order to reach a certain target configuration, the trace must go through configurations with a higher number of certain resources than the target. 
Such traces may be less useful to the designer, or might not fit with their intent. For example, consider a chemical solution of three chemical compounds $A,B,C$, whose respective quantities are modelled as $(x_A,x_B,x_C)$. Suppose we start at $(1,2,2)$ and wish to reach $(10,10,10)$. 
It may be the case that due to reactions between the materials, one first needs to flood the solution with compound $A$ (going above $10$), and only then the reactions can lead down to $(10,10,10)$. In case $A$ is expensive, dangerous or volatile, this trace might not be useful for actual reachability. Similar settings can occur with any types of resources, where one does not want to go over the amount of the resources in the target before reaching it.

To capture this notion, we introduce \emph{box-reachability}: a vector $\vec{x}=(x_1,\ldots,x_d)\in \bbN^d$ is \emph{box-reachable} in a VAS $\cV$ if there is a trace from $\vec{0}$ to $\vec{x}$ such that throughout the entire trace, each coordinate $i$ remains at most $x_i$. 

\begin{example}
\label{xmp:box reachability}
Consider a \dVAS{2} $\cV=\{(-1,2),(2,-1),(10,10)\}$. Let us examine the configurations $\vec{t_1}=(11,11)$ and $\vec{t_2}=(21,21)$.

Observe that $(-1,2)+(2,-1)=(1,1)$, and the path $(-1,2),(2,-1)$ can be taken from every vector $(x,y)\ge (1,0)$. Thus, both $\vec{t_1}$ and $\vec{t_2}$ are reachable in $\cV$, via the paths
$\pi_1: \vec{0}\Nleads{(10,10)(-1,2)(2,-1)}{(11,11)}$  and $\pi_2: \vec{0}\Nleads{(10,10)((-1,2)(2,-1))^{11}}{(21,21)}$, respectively.
However, both $\pi_1$ and $\pi_2$ exceed their respective boxes.  Indeed, the penultimate vector reached by $\pi_1$ is $(9,12)$ (exceeding $(11,11)$), and by $\pi_2$ is $(19,22)$ (exceeding $(21,21)$). 

It is not hard to see that $\vec{t_1}\notin \boxreach{\cV}$. Indeed, the first vector that must be taken in any path in $\cV$ is $(10,10)$, and from there any added vector in $\cV$ exceeds the $(11,11)$ box.

However, $\vec{t_2}\in \boxreach{\cV}$, via the run $\pi_3: \vec{0}\Nleads{(10,10)(-1,2)(2,-1)(10,10)}{(21,21)}$.
Moreover, a similar path can be taken to box-reach any vector $(k,k)$ with $k \ge20$, via the path $\vec{0}\Nleads{(10,10)((-1,2)(2,-1))^{k-20}(10,10)}{(k,k)}$. 

This example illustrates two central characteristics of the box-reachable set. First, we notice that the classical reachability set does not coincide with box-reachability. Indeed, as we show above, $(11,11)$ is reachable but not box-reachable.
An even more obvious example of this is that $(30,0)$ is reachable in $\cV$ (via the path $(10,10)(2,-1)^{10}$), but is not box-reachable, since all paths start in $(10,10)$ and in particular go above $0$ in the $\yc$-coordinate.

The second characteristic is that vectors that are ``deep'' enough in the positive quadrant, e.g., $(21,21)$, allow us enough wiggle room along the path to apply all the ``unsafe'' vectors that exceed their box (e.g., $(-1,2),(2,-1)$) without leaving the overall box. This suggests that perhaps the box-reachability set and the classical reachability set do coincide for vectors with large-enough entries. Indeed, this is the main contribution of this paper, for dimension 2.
 \hfill \qed
\end{example} 

\paragraph{Contribution.}
We focus on dimension 2 (see \cref{xmp:2vass and 3vas,rmk: one cycle is crucial} for context on this). 
In \cref{thm:main:reach coincides with box reach from W} we characterize the box-reachable set as follows: for every $\dVAS{2}$ there exists a (polynomial size) $W\in \bbN$ such that the reachability set and the box-reachability sets coincide on $[W,\infty)^2$.
%(\cref{thm:main:reach coincides with box reach from W}), and in $\bbN^2\setminus [W,\infty)^2$ the box-reachability set is effectively semilinear (\cref{thm:main:box reach is semilinear below W}).

Technically, the proof of \cref{thm:main:reach coincides with box reach from W}
uses insights on the geometry of the reachable set in the context of the cone it spans. We rely on nontrivial machinery from convex geometry (specifically, Steinitz's lemma~\cite{grinberg1980value} and the Deep-in-the-cone Lemma~\cite{cslovjecsek2025parameterized}; both are introduced in \cref{sec:cones and lattices}). 

Focusing on $\dVAS{2}$ may seem limited. To justify this, we show that our result no longer holds neither for $\dVAS{3}$ nor for $\dVASS{1}$ (i.e., Vector Addition Systems with States). 

In addition, we study the shape of the box-reachable set, showing that it is semilinear for $\dVAS{2}$ as well as $\oVASS$, and drawing a connection between box reachability in $\dVAS{d}$ and standard reachability in $\dVAS{2d}$.

%The proof of \cref{thm:main:box reach is semilinear below W} relies on the well-understood structure of the reachability set for $\dVASS{1}$ (Vector Addition System with States), where the reduction from $\dVAS{2}$ to $\dVASS{1}$ is due to the bound $W$, which essentially eliminates one of the counters. Specifically, we rely on the special structure of Linear Path Schemes (LPS) that characterize the reachability set of a $\dVASS{1}$~\cite{blondin2015,valiant1975deterministic}. Using this structure, together with careful cycle-analysis, we are able to obtain the semilinearity of the box-reachable set.

In a broader context, our work can be viewed as a form of \emph{boundedness} constraint on the VAS, in the sense that box-reaching runs are bounded by their target. Unlike boundedness by a constant, this still allows counters to become arbitrarily large. 

In terms of understanding the geometry of VAS, our contribution stands in an interesting contrast to standard methods of reasoning about VAS: the most common approach in reasoning about low-dimensional VAS (and VASS) is that of Linear Path Schemes, which intuitively allow to iterate a small number of cycles many times. However, LPS are ``highly non-box safe'', in the sense that LPS typically reach their target by taking one cycle many times, making some counter very large, and then correcting this by iterating another cycle, and so on. Thus, LPS guide the path far outside the box and then back. 
In our work, we need to exactly avoid this type of behaviors.

\paragraph{Related Work.}
Reachability in VAS and VASS have received much attention in recent years. We refer the reader to \cite{schmitz2016complexity} for a survey, and to~\cite{leroux2022reachability,czerwinski2022reachability} for more recent works and references therein.
More closely related to this work are studies concerning boundedness and geometrical understanding of the reachability set. 
Boundedness has been studied in early works on VAS~\cite{karp1969parallel,rackoff1978covering}, as it is essential to the study of coverability. Later studies~\cite{rosier1986multiparameter,demri2013selective,almagor2020coverability,kunnemann2023coverability} refine the notion of boundedness or look at variants of the model, as well as obtain improved bounds for coverability.

On the geometry front, several works have explicitly tried to provide a better understanding of the reachability sets of VASS, e.g.,~\cite{blondin2017logics,almagor2023geometry} where the geometry of the continuous variant of VASS is considered, and~\cite{guttenberg2023geometry} which explicitly aims to characterize geometric properties of the reachability set.

\paragraph{Paper Organization.}
%\shtodo{Fix at the end}
In \cref{sec:box reachability} we present the notion of box reachability, and show why we focus on $\dVAS{2}$. In \cref{sec: reach and box reach coincide} we prove that reachability and box-reachability eventually coincide. This is split into several parts: \cref{sec:cones and lattices} presents the geometric tools we use, and \cref{sec: cone contains first quadrant,sec: cone does not contain quadrant} prove the theorem by addressing two different cases, with different techniques.
In \cref{sec:semilinearity} we show the semilinearity of box-reachability for $\dVAS{2}$ and for $\oVASS$.
%\cref{sec: box reachability is semilinear near axes} then proves our second result -- the semilinearity of box-reachability near the axes. In particular, \cref{sec:box reach in 1 VASS} characterizes box-reachability in $\dVASS{1}$. 
We conclude with a discussion in \cref{sec:discussion}. 

Due to space constraints, most proofs appear in the appendix.

% \begin{enumerate}
%     %\item Motivation for reachability, and its drawbacks (resources)
%     %\item box-reachability to the rescue!
%     %\item technically easier, when considering one vector.
%     \item Context of computing with existentially bounded resources.
%     \item interesting to characterize the entire set
%     \item In 2D reachability -- LPS: highly non-box.
%     \item We do that, focus on 2D.
%     \item tools from geometry and vanilla-combinatorics.
%     \item Related: bounded VASS, continuous VASS~\cite{guttenberg2023geometry}
%     \item Can't always fix an LPS to become box-reachable.
%     \item Cite disequality, maybe param. univ.
% \end{enumerate}

\section{Preliminaries}
Let $\bbN=\{0,1,2,\dots\}$ denote the naturals and $\bbZ$ the integers. 
For $i\le j\in \bbZ$ we denote by $[i,j]=\{i,i+1,\dots,j\}$ the interval between $i$ and $j$. We extend the notation to $[i,\infty]=\{m\in\bbZ\mid m\ge i\}$. 
We denote vectors in $\bbZ^d$ in bold (e.g., $\vec{u}$), and refer to their components by subscripts: $\vec{u}=(\vec{u}_1, \ldots, \vec{u}_d)$. We denote $\norm{\vec{u}}=\max_{i\in [1,d]}\{|\vec{u}_i|\}$ the infinity-norm of $\vec{u}$.

For two vectors $\vec{u},\vec{v}\in \bbZ^d$ we write $\vec{u}\leq \vec{v}$ if $\vec{u}_i \leq \vec{v}_i$ for all $i\in [1,d]$.
For a set $S$ we denote by $S^*$ (resp. $S^+$) the set of finite sequences (resp. non-empty finite sequences) of elements of $S$.
For a vector $\vec{v}\in \bbR^2$, we typically refer to its coordinates as $\vec{v}=(\vec{v}_\xc,\vec{v}_\yc)$. 

A \emph{vector addition system of dimension $d$} ($\dVAS{d}$, for short) is a finite set of vectors $\cV\subseteq \bbZ^d$. We denote by $\norm{\cV}=d\cdot \sum_{\vec{v}\in \cV}\norm{\vec{v}}$.
A \emph{path} of $\cV$ is a finite sequence of vectors $\pi=\vec{v_1}, \vec{v_2},\ldots, \vec{v_n}\in \cV^*$ (we sometimes omit the commas for brevity). The \emph{length} of $\pi$ is $|\pi|=n$.
For $i,j\le n$ we denote the infix $\pi[i,j]=\vec{v_i},\ldots,\vec{v_j}$ (if $j<i$ the infix is empty). We also denote $\pi[j,\ldots]=\pi[j,|\pi|]$ and $\pi[j]=\pi[j,j]$.

We introduce some notation for various properties of $\pi$: the \emph{effect} of $\pi$ is $\eff(\pi)=\sum_{i=1}^n\vec{v_i}$. Fix a coordinate $1\le k\le d$, we denote by $\eff_k(\pi)=\eff(\pi)_k$ the effect in coordinate $k$. 
The \emph{drop} and \emph{peak} of $\pi$ in coordinate $k$ are $\drop_k(\pi)=|\min_{j\in [0,n]}{\eff_k(\pi[1,j])}|$ and $\peak_k(\pi)=\max_{j\in [0,n]}{\eff_k(\pi[1,j])}$. Observe that $\drop_k(\pi)\ge 0$ and $\peak_k(\pi)\ge 0$, since for $j=0$ we have that $\pi[1,0]$ is empty, and thus has effect $\vec{0}$.

A starting vector $\vec{s}$ and a path $\pi$ as above induce the \emph{trace} $\vec{s_0},\vec{s_1},\ldots,\vec{s_n}$  where $s_0=s$ and for every $0<i\le n$ we have $s_i=s_{i-1}+\vec{v_i}$. We then write $\vec{s}\leads{\pi}\vec{s_n}$. If $\vec{s_i}\in \bbN^d$ for all $i\in [0,n]$, we write $\vec{s}\Nleads{\pi}\vec{s_n}$, and if we wish to emphasize that some coordinates may be negative, we write $\vec{s}\Zleads{\pi}\vec{s_n}$.

A vector $\vec{t}$ is \emph{reachable} in $\cV$ if there exists a path $\pi$ such that $\vec{0}\Nleads{\pi}\vec{t}$. We say that $\vec{t}$ is \emph{$\bbZ$-reachable} if $\vec{0}\leads{\pi}\vec{t}$, but not necessarily via a non-negative path.
The \emph{reachability set of $\cV$} is then $\reach{\cV}=\{\vec{t}\in \bbN^d\mid \exists \pi\in \cV^* \ST \vec{0}\Nleads{\pi}\vec{t} \}$.

The fundamental problem regarding \VAS is the \emph{reachability problem}: given a \VAS $\cV$ and vector $\vec{t}$, decide whether $\vec{t}\in \reach{\cV}$. The complexity bounds on this problem were recently tightened to \Ackermannc~\cite{leroux2022reachability,czerwinski2022reachability}. 
For general dimensions $d$, the reachability set can be complicated. For dimension up to 5, however, this set is always effectively \emph{semilinear}~\cite{hopcroft1979reachability}. 
%That is, we can compute a representation of $\reach{\cV}$ as a finite union of sets of the form $\lin(\vec{b},P)=\{\vec{b}+\sum_{i=1}^k \vec{v_i}\mid \forall i\in [1,k], \vec{v_i}\in P\}$ for some $\vec{b}\in \bbN^d$ and finite $P\subseteq \bbN^d$ (these are called \emph{Linear sets}).
We remark that this is also known for $\dVASS{2}$~\cite{blondin2021reachability,leroux2004flatness}. 

\section{Box Reachability}
\label{sec:box reachability}
We start by introducing our main object of study, namely box-reachable vectors. 
Consider a \dVAS{$d$} $\cV$ and a vector $\vec{t}\in \bbN^d$. We say that $\vec{t}$ is \emph{box-reachable in $\cV$}, denoted $\vec{0}\Bleads{\pi}\vec{t}$ if there is a path $\pi\in \cV^*$ such that $\vec{0}\Nleads{\pi}\vec{t}$ and in addition, for every $1\le i\le |\pi|$ it holds that $\eff(\pi[1,i])\le \vec{t}$. That is, the trace induced by $\pi$ from $\vec{0}$ remains within the ``box'' whose opposite corners are $\vec{0}$ and $\vec{t}$.
The \emph{box-reachability set} of $\cV$ is then $\boxreach{\cV}=\left\{\vec{v}\in\bbN^d \mid \vec{v} \text{ is box-reachable in } \cV \right\}$.

In contrast with the reachability problem for \VAS, the corresponding \emph{box-reachability problem}, namely deciding whether a vector $\vec{t}$ is box-reachable in $\cV$, is far less involved. Indeed, there are at most $\norm{\vec{t}}^d$ vectors that can be traversed in order to box-reach $\vec{t}$, and since a shortest path to $\vec{t}$ does not visit the same vector twice, we can easily limit the search space. 
%\shtodo{Do we know what the complexity of $\dVAS{2}$ reachability is? (definitely PSPACE, but do we have hardness?)}
However, as discussed in \cref{sec:intro}, given a VAS $\cV$, we would like to understand the general shape of the box-reachable set of configurations. 

%\shtodo{Fix at the end}
We can now state our main result, which shows that for vectors with large-enough entries (i.e., deep enough in the first quadrant), reachability coincides with box-reachability (see \cref{sec: reach and box reach coincide}).
\begin{restatable}{theorem} {reachandboxcoincide}
    \label{thm:main:reach coincides with box reach from W}
    For every $\dVAS{2}$ $\cV$, there exists an effectively-computable $W\in \bbN$ such that $\boxreach{\cV}\cap [W,\infty]^2=\reach{\cV}\cap [W,\infty]^2$.
\end{restatable}

% Our second result addresses the points outside $[W,\infty]^2$, i.e. near the axes, as follows (see \cref{sec: box reachability is semilinear near axes}).
% \begin{restatable}{theorem}{boxreachsemilinear}
%   \label{thm:main:box reach is semilinear below W}
% %\begin{theorem}
% %    \label{thm:main:box reach is semilinear below W}
%     For every $\dVAS{2}$ $\cV$ and $W\in \bbN$, the set $\boxreach{\cV}\setminus [W,\infty]^2$ is effectively semilinear.
% %\end{theorem}
% \end{restatable}
% We remark that in order to obtain the latter, we actually show that the box-reachability for $\oVASS$ (one dimensional vector addition system \emph{with states}) is semilinear (see \cref{sec:box reach in 1 VASS}).
% %
% Since $\reach{\cV}$ is effectively semilinear~\cite{blondin2015}, then \cref{thm:main:box reach is semilinear below W,thm:main:reach coincides with box reach from W} give us the following.
% \begin{corollary}
%     \label{cor:main:box reach is semilinear}
%     Consider a $\dVAS{2}$ $\cV$, then $\boxreach{\cV}$ is effectively semilinear.
% \end{corollary}

% The proofs of the two main theorems turn out to use very different tools. Specifically, the proof of \cref{thm:main:reach coincides with box reach from W} relies on the geometry of the cone spanned by $\cV$, whereas the proof of \cref{thm:main:box reach is semilinear below W} takes a combinatorial cycle-analysis approach.

Our focus on $\dVAS{2}$ may seem restrictive. However, as we now demonstrate (see also \cref{sec:discussion}), \cref{thm:main:reach coincides with box reach from W} cannot be extended to richer models, namely to $\dVAS{3}$ or to $\dVASS{1}$ (Vector Addition Systems with States). 
\begin{example}
    \label{xmp:2vass and 3vas}
    Consider the $\dVAS{3}$ $\cV=\{(0,1,1),(1,2,-1),(1,-1,2)\}$ and targets of the form $(2n,n+1,n+1)$. Note that these targets are arbitrarily ``deep'' in the positive octant, and are reachable via the path $(0,1,1),((1,2,-1),(1,-1,2))^n$.
    % \mitodo{shouldn't this be $(2n,n+1,n+1)$ as $(1,2,-1),(1,-1,2)$ sums up to $(2,1,1)$?}
    % \shtodo{Right! Good catch.}
    % \itodo{It may be stupid but I noticed inconsistent notation: paths here are with commas and +, and in the previous example not}
    % \shtodo{Good catch. I removed $+$. I explicitly mention in the prelim that we use either commas or dots, depending on what's convenient.}

    However, these targets are not box-reachable: in order to start any path, the vector $(0,1,1)$ must be taken first.  Then, it is easy to see that the number of $(1,2,-1)$ and $(1,-1,2)$ used must be equal, and equal to $n$. Therefore, after taking $(0,1,1)$, the rest of any reaching path consists of vectors with negative entries, and in particular is not box-reaching.

    For $\dVASS{1}$, it suffices to have two states $q_0,q_1$ such that $q_1$ is only reachable via a negative number (but is reachable). Then any path ending in $q_1$ cannot be box-reaching.
\end{example}

\section{Reachability and Box Reachability Eventually Coincide}
\label{sec: reach and box reach coincide}
In this section we prove \cref{thm:main:reach coincides with box reach from W}. The proof relies on two results regarding cones and lattices: the first is the \emph{Steinitz Lemma}~\cite{grinberg1980value}, which intuitively allows us to reorder paths so that they do not diverge wildly, and instead stay within some bounded ``corridor''. The second is the \emph{Deep-in-the-Cone Lemma}~\cite{cslovjecsek2025parameterized}, which connects $\bbN$-reachability with $\bbZ$-reachability. 
We start with some definitions and known results.

\subsection{Reachability, Cones and Lattices}
\label{sec:cones and lattices}
Consider a set of vectors $D=\{\vec{v_1}, \vec{v_2},\ldots,\vec{v_n}\}\subseteq\bbZ^d$ (we do not think of it as a \VAS at this point). 
The \emph{cone} spanned by $D$ is the set of vectors in $\bbR^d$ expressible as non-negative combinations of the vectors in $D$. Similarly, we define the \emph{integer cone} where we restrict attention to non-negative integer combinations, and the \emph{lattice}, where the coefficients are all integers:
% \[
% \cone(D) := \left\{ \sum_{i=1}^n \lambda_i \vec{v_i} \ \Big\vert \ \lambda_i \in \bbR_{\geq 0} \right\} \subseteq \bbR^d
% \]
\begin{align*}
    &\cone(D) := \left\{ \lambda_1 \vec{v_1}+\ldots +\lambda_n\vec{v_n} \ \vert \ \forall i.\ \lambda_i \in \bbR_{\geq 0} \right\} \subseteq \bbR^d\\
    &\intCone(D) := \left\{ \lambda_1 \vec{v_1}+\ldots +\lambda_n\vec{v_n} \ \vert \ \forall i.\ \lambda_i \in \bbN \right\} \subseteq \bbZ^d\\
    &\lattice(D) := \left\{ \lambda_1 \vec{v_1}+\ldots +\lambda_n\vec{v_n} \ \vert \ \forall i.\ \lambda_i \in \bbZ \right\} \subseteq \bbZ^d
\end{align*}
The set $D$ can be thought of as a representation of $\cone(D)$, typically called the \emph{V-representation} (where V stands for ``Vertex''). A classical result by Weyl~\cite{weyl1934elementare} shows that we can compute from $D$ a set $\cF\subseteq\bbZ^d$ such that\footnote{$\tup{\cdot,\cdot}$ is the standard inner product.}
$
\cone(D) = \left\{ \vec{v}\in\bbR^d \mid \forall \vec{f}\in\cF: \tup{\vec{f},\vec{v}}\ge 0 \right\}.$
The set $\cF$ is called the \emph{H-representation} of $\cone(D)$ (where H stands for ``Halfspace'').
Intuitively, the set $\cF$ can be understood as comprising vectors that are perpendicular to the boundaries of the cone. Each vector $\vec{f}\in\cF$ defines a (positive) halfspace in $\bbR^d$, and the inequality $\tup{\vec{f},\vec{v}}\geq 0$ states that the vector $\vec{v}$ lies on the ``correct'' side of this hyperplane. 

In general, computing the H-representation from $D$ may take exponential time. 
However, for cones in $\bbR^2$ an H-representation has at most two vectors $\cF=\{\vec{f_1},\vec{f_2}\}$, and it can be computed in polynomial time. 
Moreover, $\norm{\vec{f_1}}$ and $\norm{\vec{f_2}}$ are polynomial in $\norm{D}=\max\{\norm{\vec{v}}\mid \vec{v}\in D\}$. Indeed, this follows by finding two ``extremal'' vectors that span the cone (or by first detecting that the cone is all of $\bbR^2$, a single halfspace, or a single ray -- easier cases commented on in the following), which is in turn done by sorting $v_1,\ldots, v_n$ by their angle with the $\xc$-axis, and then the facets are defined by normals to these two generating vectors. In particular, the entries of the normals are integers, and their description size is identical to that of the vectors, e.g., the normal we take to $(x,y)$ is either $(-y,x)$ or $(y,-x)$, depending on the required direction of the halfspace.
%\shtodo{Maybe add a short appendix about this.}
%\shtodo{cite or prove}

For an H-representation $\cF$ of $\cone(D)$, each $\vec{f}\in \cF$ induces a measure of how ``far'' a vector $\vec{v}$ is from the respective facet of $\cone(D)$, namely the inner product $\dotp{\vec{f},\vec{v}
}$.
Specifically, given some $M\in \bbN$, we say that a vector $\vec{v}\in \mathrm{cone}(D)$ is \emph{$M$-deep in the cone} if $\dotp{\vec{f},\vec{v}}\geq M$ for every $\vec{f}\in\cF$.
Our first tool is the following lemma from~\cite{cslovjecsek2025parameterized} (which actually follows from existing literature, see references in~\cite{cslovjecsek2025parameterized}), which connects reachability in the integer cone with reachability in the lattice.
\begin{lemma}[Deep-in-the-Cone, Lemma 16 in~\cite{cslovjecsek2025parameterized}]
\label{lem:ditc}
Given a set of vectors $D=\{\vec{v_1}, \vec{v_2},\ldots,\vec{v_n}\}\subseteq\bbZ^d$ and the H-representation $\cF$ of $\cone(D)$, there exists a constant $M\in \bbN$, depending on $D$ and $\cF$, such that if $\vec{v}\in \cone(D)\cap\bbZ^d$ is $M$-deep in the cone, then $\vec{v}\in \intCone(D)$ if and only if $\vec{v}\in \lattice(D)$. 

Moreover, $M$ is polynomial in $\norm{D},\norm{\cF}$ and exponential in $d$.
\end{lemma}
Since in our setting the dimension is $d=2$, and as mentioned above, $\norm{\cF}=\norm{\cD}$, we have that $M$ in \cref{lem:ditc} is polynomial in $\norm{D}$.

Our second tool is the classical \emph{Steinitz Lemma}.
\begin{lemma}[Steinitz, as stated in~\cite{grinberg1980value}]
\label{lem:steinitz}
    Let $\vec{v_1},\dots,\vec{v_k}$ be a non-empty sequence of vectors in $\bbR^d$.
    Let $\vec{v}=\sum_{j=1}^k \vec{v_j}$ and $I=\max_{1\le j\le k}\norm{\vec{v_j}}_\infty$. There exists a permutation $\sigma$ of $\{1,\dots,k\}$ such that for every $n\in\{d,\dots,k\}$, we have:
    \[\norm{\sum_{j=1}^n \vec{v_{\sigma(j)}} - \frac{n-d}{k}\vec{v}}_\infty \le d\cdot I\]
\end{lemma}
In the context of \dVAS{2}, this lemma states that if a vector $\vec{v}$ is reachable via some path $\pi$, then the vectors of $\pi$ can be rearranged so that the resulting path, dubbed a \emph{Steinitz path} does not stray too far from the straight line that connects the origin to $\vec{v}$. Specifically, the guaranteed bound (dubbed the \emph{Steinitz constant}) is $d\cdot I=2\norm{\cV}$. 
The form of Steinitz paths guarantees in particular that their drop and peak are not too large, which is particularly useful for us. 
More precisely, we have the following (see \cref{apx: drop and peak of steinitz} for the proof).
\begin{lemma}
    \label{lem: drop and peak of steinitz}
    Consider a Steinitz path $\pi$ with $\eff(\pi)=(x,y)\ge (0,0)$, then $\drop_\xc(\pi),\drop_\yc(\pi)\le 2\norm{\cV}$ and $\peak_\xc(\pi)\le x+2\norm{\cV}$ and $\peak_\yc(\pi)\le y+2\norm{\cV}$.
\end{lemma}
% \begin{proof}
% Let $\pi$ be a Steinitz path. We prove the bounds for $\xc$, the bounds  for $\yc$ are obtained identically.

% We start with $\drop_\xc(\pi)$.
% Note that $\effx(\pi[1,1])\ge -2\norm{\cV}$, since the maximal effect of a single transition is $\norm{\cV}$. 
% Consider $2\le i\le |\rho'|$, then we can use the properties of $\pi$ as a Steinitz path. 
% First, if $\effx(\pi[1,i])\ge \frac{i-2}{|\pi|}x$, then in particular $\effx(\pi[1,i])\ge 0\ge -2\norm{\cV}$ and we are done.
% Otherwise, we have 
% \[
% \frac{i-2}{|\pi|}x- \effx(\pi[1,i])\le 2\norm{\cV} \implies \effx(\pi[1,i])\ge -2\norm{\cV}. 
% \]
% so we have $\drop_\xc(\pi)\le 2\norm{\cV}$. 

% We proceed with $\peak_\xc(\pi)$.
% Let $1\le i\le |\pi|$. Note that for $i=1$ we have $\effx(\pi[1,1])\le \norm{\cV}\le x+2\norm{\cV}$.
% Thus, we can focus on $i\ge 2$.
% If $\effx(\pi[1,i])\le \frac{i-2}{\pi}x$, then $\effx(\pi[1,i])\le x\le x+2\norm{\cV}$ and we are done.
% Otherwise, we use the Steinitz property: 
% \[\begin{split}
% \effx(\pi[1,i])-\frac{i-2}{\pi}x\le 2\norm{\cV} \implies
% \effx(\pi[1,i])\le 2\norm{\cV}+\frac{i-2}{\pi}x\le x+ 2\norm{\cV}
% \end{split}
% \]
% so we have $\peak_\xc(\pi)\le x+2\norm{\cV}$. 
% \end{proof}

\subsection{Proof of \cref{thm:main:reach coincides with box reach from W}}
We are now ready for the main result, which we prove in the remainder of this section.
\reachandboxcoincide*

For $W\in \bbN$, denote the ``L-shape'' $\lshape=([0,W-1]\times \bbN) \cup (\bbN\times [0,W-1])$ and its complement $\clshape=[W,\infty]^2$. Also denote the real positive quadrant by $\posquad$.
Note that $\boxreach{\cV} \subseteq \reach{\cV}$ trivially holds, and this inclusion remains valid when restricted to $\clshape$. Therefore, it suffices to demonstrate the reverse inclusion: $\reach{\cV}\cap \clshape \subseteq \boxreach{\cV}\cap \clshape$.

Fix a $\dVAS{2}$ $\cV=\{\vec{v_1},\dots,\vec{v_k}\}$, and denote by $M$ the Deep-in-the-Cone constant guaranteed by \cref{lem:ditc}.
We start by filtering out some degenerate cases, which would allow us to make some simplifying assumptions on $\cV$.

We can assume without loss of generality that not all vectors in $\cV$ are of the form $(\leq 0, \leq 0)$ (i.e., $(x,y)$ such that $x\le 0,y\le 0$). Indeed, if this were the case, then $\reach{\cV}=\{\vec{0}\}$ and the theorem trivially holds. Similarly, it cannot hold that every vector $(x,y)\in \cV$ has a negative coordinate, as again we would have $\reach{\cV}=\{\vec{0}\}$.
Our next assumption is that the vectors in $\cV$ are not all linearly dependent. Indeed, if this is the case then $\reach{\cV}$ is one-dimensional, and the setting is simpler: in this case all vectors in $\cV$ are scalar multiplications of some vector $(x,y)\in \bbN^2$. Then, reachability and box-reachability can be reasoned about in one of the components, so the setting is that of $\dVAS{1}$.
% \mitodo{Is this reasoning in one of the components or transforming to a 1-VAS? I see this as saying $(x,y)=1$ and then using 1-VAS, but not as ignoring one of the dimensions}
We handle this in \cref{apx: the 1 VAS case}.

Having set these assumptions, our first step is to find a vector reachable in $\cV$ that has strictly positive coordinates, i.e., of the form \((>0, >0)\), and that is box reachable. 
Note that if $\cV$ has a single generating vector of the form $(>0,>0)$, then it satisfies this requirement.
If no generator of $\cV$ is of this form, then there must be vectors in $\cV$ of the form $(>0, 0)$ or $(0, >0)$ (otherwise we have $\reach{\cV}=\{\vec{0}\}$ as shown above). 
Without loss of generality, assume there is a vector \(\vec{u_1} = (x,0)\) where $x>0$. Since not all vectors in $\cV$ are linearly dependent, there exists a vector $\vec{u_2}=(x',y')$ with $y'>0$. Indeed, otherwise all such vectors have $y'<0$), so they cannot be taken after $(x,0)$, and can therefore be discarded from $\cV$ without changing the reachability set.

Since we assume there is no $(>0,>0)$ vector in $\cV$, it follows that $x'\le 0$. We now define $\vec{s}=(-2x'+1)\vec{u_1}+\vec{u_2}=((-2x'+1)x+x',y')=(x+x'(1-2x),y')$, and notice that $x+x'(1-2x)>0$, since $1-2x< 0$ and $x'\le 0$. Thus, $\vec{s}$ is of the form $(>0,>0)$. Moreover, we claim that $\vec{s}$ is box-reachable. 
Indeed, consider the path $\zeta=(x,0)^{-x'}(x',y')(x,0)^{-x'+1}$ then $\eff(\zeta)=\vec{s}$. Also, $\drop_\xc(\zeta)=\drop_\yc(\zeta)=0$ since the only negative coordinate is $x'$, and that is taken only after $(x,0)^{-x'}$. Finally, $\peak_\yc(\zeta)=y'=\effy(\zeta)$ and $\peak_\xc(\zeta)=-x'x+x'+(-x'x)=\effx(\zeta)$. 

%As $\reach{\cV}$ is \shcomm{not one-dimensional}{did we discuss the one-dimensional case? Where?}, 
%\itodo{We did not write it. Is it a special case of Michael's work? or should we write it here?}
%\shtodo{It's not exactly a special case. We need to somehow show that we can make it into a special case, perhaps by taking a common denominator.}
%there must also exist a vector $u_2=(y,z)$ where $y\in\bbZ^-$ and $z\in\bbN^+$. 
%We now define the vector $s=-y\cdot u_1+u_2$, and let 
We now define $\spos=2\norm{\cV}\vec{s}$.
Notice that $\norm{\spos}=2\norm{\cV}\norm{((-2x'+1)x+x',y')}\le 2\norm{\cV}(2\norm{\cV}^2+2\norm{\cV})\le 8\norm{\cV}^3$. That is, the representation of $\spos$ is polynomial in $\norm{\cV}$. Moreover, we have $\vec{s}\ge (1,1)$ as it is positive and integer. Therefore, $\spos\ge (2\norm{\cV},2\norm{\cV})$. Also, since $\vec{s}$ is box-reachable, so is $\spos$. We summarize the properties of $\spos$ for later reference.
\begin{observation}
    \label{obs:spos norm} 
    $\spos$ is box-reachable in $\cV$, satisfies $\spos\ge (2\norm{\cV},2\norm{\cV})$ and 
    $\norm{\spos}\le 8\norm{\cV}^3$.
\end{observation}

In the following we split our analysis to two cases, according to whether $\cone(\cV)$ contains the positive quadrant $\posquad$ or not. The cases are depicted in \cref{fig:cone and quadrant}. 
Specifically, since we assume that $\cone(\cV)$ is not one-dimensional (\cref{apx: the 1 VAS case}) or all of $\bbR^2$ (\cref{rmk:cone is R2}), we have that $\cone(\cV)=\cone(\vec{\chi_1},\vec{\chi_2})$, where $\vec{\chi_1},\vec{\chi_2}$ are the two extremal vectors whose cone spans all the vectors in $\cone(\cV)$, as discussed in \cref{sec:cones and lattices} (indeed, a nontrivial cone in $\bbR^2$ is spanned by at most two vectors). 
%Note that the angle between $\vec{\chi_1}$ and $\vec{\chi_2}$ is at most $180^\circ$ (since the cone comprises positive combinations).

We can then formulate the different cases as follows: 
\begin{itemize}
    \item If $\vec{\chi_1}$ is of the form $(\le 0,\ge 0)$ and $\vec{\chi_2}$ is of the form $(\ge 0,\le 0)$, then $\posquad\subseteq \cone(\cV)$ (depicted in \cref{fig:cone contains quadrant}). 
    \item If $\vec{\chi_1},\vec{\chi_2}$ are both of the form $(\ge 0,\ge 0)$, then $\cone(\cV)\subseteq \posquad$ (depicted in \cref{fig:quadrant contains the cone}).
    \item If $\vec{\chi_1}$ is of the form $(> 0,>0)$ and $\vec{\chi_2}$ is of the form $(\bbZ ,\le 0)$ (or vice-versa), then $\posquad\cap \cone(\cV)\neq \emptyset$ (depicted in \cref{fig:cone contains x axis,fig:cone contains y axis}).
\end{itemize}
Importantly, there are no other cases to consider. Indeed, keeping in mind that the angle between $\vec{\chi_1}$ and $\vec{\chi_2}$ is at most $180^\circ$ (since the cone comprises positive combinations), a simple examination of the other configurations of $\vec{\chi_1}$ and $\vec{\chi_2}$ shows that they lead to either $\cone(\cV)\cap \posquad=\{0\}$ %or to $\cone(D)\cap \posquad=\bbR^2$, 
or to a one dimensional cone, which we already handled.
% \mitodo{Why does having positive combinations means the angle can't be more than 180? Why can't we have two vectors that have an angle of more than 180 between them?}

\begin{figure}[ht]
    \centering
    \captionsetup{justification=centering}
    % \shtodo{add $\vec{\chi_1},\vec{\chi_2}$ to the figures.}
    \begin{subfigure}{0.24\linewidth}
        \begin{tikzpicture}[scale=0.4]
            \draw[fill=gray!40, color=gray!40] (0, 0)--(-2, 5)--(5,5)--(5, -2)--(0, 0);
            % \draw[help lines, color=gray!30, dashed] (-1.9,-1.9) grid (4.9,4.9);
            \draw[->, thick] (-2,0)--(5,0) node[right]{$x$};
            \draw[->, thick] (0,-2)--(0,5) node[above]{$y$};
            \draw[ultra thick,color=black] (0, 0)--(-2, 5);
            \draw[->,thick,color=black] (0, 0)->(-1, 2.5) node[left, pos=0.6] {$\vec{\chi_1}$};
            \draw[ultra thick,color=black] (0, 0)--(5,-2);
            \draw[->,thick,color=black] (0, 0)->(2.5, -1) node[below, pos=0.6] {$\vec{\chi_2}$};
        \end{tikzpicture}
        \caption{}
        \label{fig:cone contains quadrant}
    \end{subfigure}
    \begin{subfigure}{0.24\linewidth}
        \begin{tikzpicture}[scale=0.4]
            \draw[fill=gray!40, color=gray!40] (0, 0)--(2, 5)--(5,5)--(5, 3)--(0, 0);
            % \draw[help lines, color=gray!30, dashed] (-1.9,-1.9) grid (4.9,4.9);
            \draw[->, thick] (-2,0)--(5,0) node[right]{$x$};
            \draw[->, thick] (0,-2)--(0,5) node[above]{$y$};
            \draw[ultra thick,color=black] (0, 0)--(2, 5);
            \draw[->,thick,color=black] (0, 0)->(1, 2.5) node[pos=1.3] {$\vec{\chi_1}\quad\:\ $};
            \draw[ultra thick,color=black] (0, 0)--(5,3);
            \draw[->,thick,color=black] (0, 0)->(2.5, 1.5) node[below, pos=0.8] {$\quad\vec{\chi_2}$};
        \end{tikzpicture}
        \caption{}
        \label{fig:quadrant contains the cone}
    \end{subfigure}
    \begin{subfigure}{0.24\linewidth}
        \begin{tikzpicture}[scale=0.4]
            \draw[fill=gray!40, color=gray!40] (0, 0)--(-2, 5)--(5,5)--(5, 3)--(0, 0);
            % \draw[help lines, color=gray!30, dashed] (-1.9,-1.9) grid (4.9,4.9);
            \draw[->, thick] (-2,0)--(5,0) node[right]{$x$};
            \draw[->, thick] (0,-2)--(0,5) node[above]{$y$};
            \draw[ultra thick,color=black] (0, 0)--(-2, 5);
            \draw[->,thick,color=black] (0, 0)->(-1, 2.5) node[left, pos=0.6] {$\vec{\chi_2}$};
            \draw[ultra thick,color=black] (0, 0)--(5,3);
            \draw[->,thick,color=black] (0, 0)->(2.5, 1.5) node[below, pos=0.8] {$\quad\vec{\chi_1}$};
        \end{tikzpicture}
        \caption{}
        \label{fig:cone contains y axis}
    \end{subfigure}
    \begin{subfigure}{0.24\linewidth}
        \begin{tikzpicture}[scale=0.4]
            \draw[fill=gray!40, color=gray!40] (0, 0)--(2, 5)--(5,5)--(5, -2)--(0, 0);
            % \draw[help lines, color=gray!30, dashed] (-1.9,-1.9) grid (4.9,4.9);
            \draw[->, thick] (-2,0)--(5,0) node[right]{$x$};
            \draw[->, thick] (0,-2)--(0,5) node[above]{$y$};
            \draw[ultra thick,color=black] (0, 0)--(2, 5);
            \draw[->,thick,color=black] (0, 0)->(1, 2.5) node[pos=1.3] {$\vec{\chi_1}\quad\:\ $};
            \draw[ultra thick,color=black] (0, 0)--(5,-2);
            \draw[->,thick,color=black] (0, 0)->(2.5, -1) node[below, pos=0.6] {$\vec{\chi_2}$};
        \end{tikzpicture}
        \caption{}
        \label{fig:cone contains x axis}
    \end{subfigure}
    \caption{Possible overlaps of the cone and the first quadrant. In \cref{fig:cone contains quadrant} the cone contains the positive quadrant, while in \cref{fig:quadrant contains the cone,fig:cone contains y axis,fig:cone contains x axis} it does not. Note that Cases (c) and (d) may intersect the \emph{negative} quadrant, even though this is not explicitly depicted.}
    \label{fig:cone and quadrant}
\end{figure}
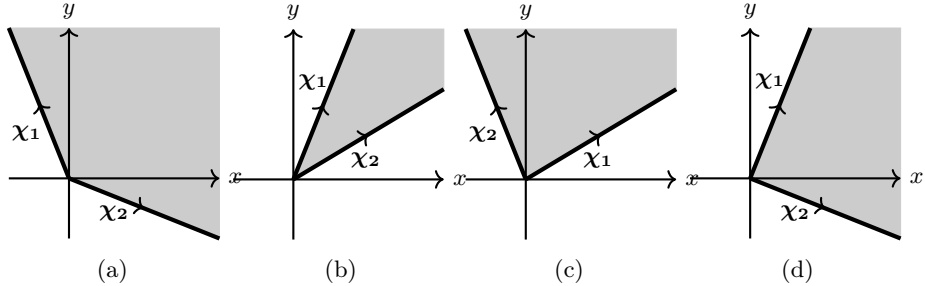

\subsubsection{First Case: The Cone Contains the Positive Quadrant}
\label{sec: cone contains first quadrant}
We present a sketch of the proof in this case (see \cref{fig:first case cone contains pos quad} for a depiction). The full details appear in \cref{apx:proof of first case}.

Recall the box-reachable, strictly positive vector $\spos$ in $\cV$ (as per \cref{obs:spos norm}). We choose $W$ large enough so that any target $\vec{v}\in \clshape$ is so far from the axes that $\vec{v}-2\spos$ is $M$-deep-in-the-cone (for the appropriate $M$ obtained from \cref{lem:ditc}). Note that this relies on the fact that the cone contains the first quadrant, so that being deep in the first quadrant also implies being deep-in-the-cone.
Intuitively, the vector $\vec{v}-2\spos$ captures the fact that we consider a path to $\vec{v}$ starting and ending with $\spos$.

Assume that $\vec{v}$ is $\bbN$-reachable in $\cV$. In particular, $\vec{v}\in \intCone(\cV)$. By adding $-2\spos$ to $\vec{v}$ (note the negative coefficient $-2$) we have that $\vec{v}-2\spos\in \lattice(\cV)$. 
We can now invoke \cref{lem:ditc} and get that $\vec{v}-2\spos\in \intCone(\cV)$.
Then, we obtain a Steinitz path $\rho$ with effect $\vec{v}-2\spos$ as per \cref{lem:steinitz}. By \cref{lem: drop and peak of steinitz} and our careful choice of $\spos$, we can show that the path $\spos\cdot \rho\cdot\spos$ box reaches $\vec{v}$. Intuitively, this is because $\spos$ gets us far enough from the origin to overcome the Steinitz corridor of $\rho$, and since the effect of $\spos\cdot \rho$ is $\vec{v}-\spos$, then this corridor also does not exceed $\vec{v}$, allowing us to complete the path with another $\spos$.
\begin{remark}
    \label{rmk:cone is R2}
    If $\cone(\cV)=\bbR^2$, the H-representation has $\cF=\emptyset$. Thus, all vectors are deep-in-the-cone, and by \cref{lem:ditc} we have $\lattice(\cV)=\intCone(\cV)$. Therefore, this case is also covered by the analysis above.
    Similarly, if $\cone(\cV)$ is a single half-space (e.g., spanned by $(1,0),(-1,0)$, and $(0,1)$) then $\cF$ is a singleton, and the argument proceeds similarly.
\end{remark}

\begin{figure}[ht]
\centering
\begin{minipage}{.45\textwidth}
  \centering
\begin{tikzpicture}[scale=0.6]
    % Draw grid and axes
    %\draw[step=0.7cm,gray!40,very thin] (-1.6,-2) grid (8,8);
    \draw[->] (-1.5,0) -- (8,0) node[right] {$\xc$};
    \draw[->] (0,-0.5) -- (0,8) node[above] {$\yc$};

    \draw[-, line width=0.6mm, red] (0,0) -- (-1.5,8) node[midway, above left] {$\vec{\chi_2}$};
    \draw[-, line width=0.6mm, red] (0,0) -- (8,-1) node[pos=0.6, below right] {$\vec{\chi_1}$};

    % Draw the W lines
    \def\W{5.6}
    \fill [orange!25!white] (\W,\W) rectangle (8,8);
    \draw[-, thick] (\W,\W) -- (8,\W);
    \draw[-, thick] (\W,\W) -- (\W,8);
    \draw[<->, thick] (0,\W+1.4) -- (\W,\W+1.4) node[midway, above] {$W$};
    \draw[<->, thick] (\W+1.4,0) -- (\W+1.4,\W) node[midway, right] {$W$};

    % Define coordinates
    \coordinate (O) at (0,0);
    \def\sx{1.2}
    \def\sy{1.5}
    \coordinate (s) at (\sx,\sy);
    \def\vx{7}
    \def\vy{6}
    \coordinate (v) at (\vx,\vy); 
    % vmins was originally defined as v - s.
    \coordinate (vmins) at ($(v)-(s)$); 

    % Rectangle half-width
    \def\thick{1.4} 

    % Compute length and angle for the rotated coordinate system (used for the center segment)
    \pgfmathsetmacro{\L}{sqrt((\vx-\sx-\sx)^2+(\vy-\sy-\sy)^2)}
    \pgfmathsetmacro{\angle}{atan2(\vy-\sy-\sy, \vx-\sx-\sx)}

    % Draw blue vectors
    \draw[->, thick, blue] (O) -- (s) node[pos=0.8,left] {$\spos$};
    \draw[->, thick, blue] (vmins) -- (v) node[pos=0.2, right] {$\spos$};

    % Mark the point v with a small filled circle and a label.
    \filldraw[black] (v) circle (2pt);
    \node[above right] at (v) {$\vec{v}$};

    % rectangle and the smooth, curly "wild" arrow inside it
    \begin{scope}[shift={(s)}, rotate=\angle]
        \draw[thick, black] (0,-\thick) rectangle (\L,\thick);
    
        \draw[->, thick, green!75!black, line join=round, line cap=round]
            plot[smooth] coordinates { 
                (0,0)
                (0.2*\L, 0.7*\thick)
                (0.2*\L, -0.9*\thick)
                (0.4*\L, 0.7*\thick)
                (0.5*\L, -0.5*\thick)
                (0.8*\L, 0.9*\thick)
                (0.85*\L, -0.3*\thick)
                (\L,0)
            };
    \end{scope}
\end{tikzpicture}
    \caption{The cone contains $\posquad$. The faces of the cone are depicted in red. The target $\vec{v}$ is inside $\clshape$ (highlighted orange). The blue vectors are $\spos$, and are connected by the green Steinitz path, surrounded by its bounding ``corridor''.}
    \label{fig:first case cone contains pos quad}
\end{minipage}%
\qquad
\begin{minipage}{.45\textwidth}
  \centering
\begin{tikzpicture}[scale=0.6]
    % Draw grid and axes
    % \draw[step=0.7cm,gray!40,very thin] (-1.6,-2) grid (8,8);
    \draw[->] (-1.5,0) -- (8,0) node[right] {$\xc$};
    \draw[->] (0,-0.5) -- (0,8) node[above] {$\yc$};

    % Draw the W lines
    \def\W{5.6}
    \fill [orange!25!white] (\W,\W) rectangle (8,8);
    \draw[-, thick] (\W,\W) -- (8,\W);
    \draw[-, thick] (\W,\W) -- (\W,8);
    \draw[<->, thick] (0,\W+1.4) -- (\W,\W+1.4) node[midway, above] {$W$};
    \draw[<->, thick] (\W+1.4,0) -- (\W+1.4,\W) node[pos=0.25, right] {$W$};

    % Define coordinates
    \coordinate (O) at (0,0);
    \def\sx{8/6}
    \def\sy{3/6}
    \coordinate (s) at (\sx,\sy);
    \coordinate (s2) at (2*\sx,2*\sy);
    \def\vx{7}
    \def\vy{6}
    \coordinate (v) at (\vx,\vy); 
    % vmins was originally defined as v - s.
    \coordinate (vmins) at ($(v)-(s)$);
    \coordinate (vmins2) at ($(v)-(s)-(s)$);

    % Rectangle half-width
    \def\thick{1.4} 

    % Compute length and angle for the rotated coordinate system (used for the center segment)
    \pgfmathsetmacro{\L}{sqrt((\vx-4*\sx)^2+(\vy-4*\sy)^2)}
    \pgfmathsetmacro{\angle}{atan2(\vy-4*\sy, \vx-4*\sx)}

    % Mark the point v with a small filled circle and a label.
    \filldraw[black] (v) circle (2pt);
    \node[above right] at (v) {$\vec{v}$};

    % Define the number of additional parallel lines
\def\nLines{4}
\coordinate (cap) at (-81/137,432/137);

% Compute the step size for each line (fraction of total displacement)
\foreach \i in {1,...,\nLines} {
    % Compute the fraction position
    \pgfmathsetmacro{\frac}{\i / (\nLines+1)}
    
    % Compute the new start and end points for the parallel line
    \coordinate (start) at ($(0,0) !\frac! (cap)$);  % Interpolated start
    \coordinate (end) at ($(8,3) !\frac! (v)$);        % Interpolated end
    
    % Draw the parallel line
    \draw[dashed, black] (start) -- (end);
}

\draw[-, line width=0.6mm, red] (0,0) -- (-1.5,8) node[midway, above left] {$\vhigh$};
 \draw[-, line width=0.6mm, red] (0,0) -- (8,3) node[pos=0.7, below right] {$\vlow$};

     % Draw blue vectors
    \draw[->, thick, blue] (O) -- (s) node[midway, left] {};
    \draw[->, thick, blue] (s) -- (s2) node[midway, left] {};
    \draw[->, thick, blue] (vmins) -- (v) node[midway, right] {};
    \draw[->, thick, blue] (vmins2) -- (vmins) node[midway, right] {};

    % rectangle
    \begin{scope}[shift={(s2)}, rotate=\angle]
        % \draw[thick, black] (0,-\thick) rectangle (\L,\thick);
        \draw[-,thick ,dashed] (0,-\thick)--(0,0) {};
        
        \draw[-,thick] (\L,-\thick)--(\L,\thick) {};
        
        \draw[-,thick] (0,0)--(0,\thick) {};

        \draw[-,thick] (0,\thick)--(\L,\thick) {};

        \draw[-,thick,dashed] (0,-\thick)--(1.2,-\thick) {};

        \draw[-,thick] (1.35,-\thick)--(\L,-\thick) {};
        
        \draw[->, thick, green!75!black, line join=round, line cap=round]
    plot[smooth] coordinates { 
        (0,0)                     % Start
        (0.15*\L, 0.16*\thick)     % Small deviation NW
        (0.3*\L, -0.34)            % Back parallel to s
        (0.45*\L, 0.43*\thick)    % Larger deviation NW
        (0.6*\L, 0.1)            % Back parallel to s
        (0.75*\L, 0.7*\thick)    % Even larger deviation NW
        % (0.9*\L, 0.2*\thick)     % Small return
        (\L,0)                   % End at v-s
    };
    \end{scope}
\end{tikzpicture}
    \caption{The cone intersects $\posquad$. The faces of the cone are depicted in red. The target $\vec{v}$ is inside $\clshape$ (highlighted orange). The blue vectors are $\spos$, and are connected by the green Steinitz path, surrounded by its bounding ``corridor''. The dashed lines represent the ``levels'' of dot product with $\flow$.}
    \label{fig:second case cone intersects pos quad}
\end{minipage}
\end{figure}

\subsubsection{Second Case: The Cone Does Not Contain the Positive Quadrant}
\label{sec: cone does not contain quadrant}
We now consider the setting where $\cone(\cV)$ does not contain $\posquad$. Since $\spos\in \cone(\cV)$, then $\cone(\cV)\cap \posquad\supsetneq \{\vec{0}\}$, and therefore we fall into one of the cases depicted in \cref{fig:cone contains x axis,fig:cone contains y axis,fig:quadrant contains the cone}.
Note, however, that case \cref{fig:quadrant contains the cone} is trivial, as every vector in $\cV$ has non-negative coordinates, so $\reach{\cV} = \boxreach{\cV}$. We are therefore left with the cases where $\cone(\cV)=\cone(\vec{\chi_1},\vec{\chi_2})$ neither contains or is contained in the first quadrant. That is, we have that $\vec{\chi_1}$ and $\vec{\chi_2}$ are either of the form $(>0,>0) $ and $(\bbZ,\le 0)$ (\cref{fig:cone contains x axis}), or of the form $(>0,>0) $ and $(\le 0,\bbZ)$ (\cref{fig:cone contains y axis}).

%We start with the simplest case, namely \cref{fig:quadrant contains the cone}, in which the cone is contained in $\posquad$. In this case, every vector in $\cV$ must have non-negative coordinates.
%Indeed, if there is a vector $\vec{v}_-$ with a negative component, then, $\vec{v}_-\in \cone(\cV)$ but $\vec{v}_-\notin\posquad$, so the cone would not be contained in $\posquad$.

%Therefore, every run remains in the box of its target, as exceeding the box requires applying at least one vector with a negative coordinate, which does not exist in the \VAS. Hence, the set of reachable vectors coincides with the set of box-reachable vectors, i.e., 
%$\reach{\cV} = \boxreach{\cV}$, and the claim follows immediately.

We observe that these cases are symmetric in the sense that swapping the roles of the $\xc$ and $\yc$ axes interchanges between them. Therefore, it suffices to analyze subcase \ref{fig:cone contains y axis}, i.e., where $\cone(\cV)$ contains the $\yc$-axis.
%; the argument for subcase \ref{fig:cone contains x axis} is identical and indeed can be obtained by swapping the $\xc$ and $\yc$ axes and applying our results.
We denote the extremal vectors as $\vlow=(x_1,y_1)$ with $x_1,y_1>0$ and $\vhigh=(x_2,y_2)$ 
with $x_2\le 0,y_2\in \bbZ$. The respective H-representation is then $\cF=\{\flow,\fhigh\}$ with $\flow=(-y_1,x_1)$ and $\fhigh=(y_2,-x_2)$.
We again sketch the proof, with the details in \cref{apx:proof of second case}.

Set $W$ large enough, and consider a vector $\vec{v}\in \reach{\cV}\cap \clshape$. We need to prove that $\vec{v}\in \boxreach{\cV}$. As before, we focus on $\vec{v}-2\spos$.

If we are lucky, then $\vec{v}-2\spos$ is $M$-deep-in-the-cone (spanned by $\cV$). In this case
%Our first approach is to use the same reasoning as in \cref{sec: cone contains first quadrant}: we consider $\vec{v}-2\spos$, and ask whether it is $M$-deep-in-the-cone. If it is, 
we apply the same reasoning as in the first case: we first use the lattice-reachability of $\vec{v}-2\spos$ to conclude that $\vec{v}-2\spos\in \intCone(\cV)$. Then, we build the Steinitz path $\rho'$ from $\spos$ to $\vec{v}-\spos$, and finally show that the path $\spos\cdot \rho'\cdot \spos$ box-reaches $\vec{v}$. %Note that the analysis here is conceptually slightly simpler: since the halfspace defined by $\flow$ is already above the $\xc$-axis, there is no danger of the Steinitz path reaching negative $\yc$. 
%\mitodo{What does it mean that it is above that x-axis?}
%Technically, however, the analysis is similar.

Unfortunately, unlike the first case, $\vec{v}-2\spos$ is not necessarily $M$-deep-in-the-cone. Indeed, the facet of the cone defined by 
$\vlow$ intersects $\clshape$, so there are reachable vectors that are on the facet itself, and are therefore not at all deep in the cone. Thus, we take a different approach (see \cref{fig:second case cone intersects pos quad} for a depiction).

If $\vec{v}-2\spos$ is not deep in the cone, we start by showing that it is close to the facet $\vlow$ (which is not completely trivial, since the depiction in \cref{fig:cone contains y axis} is misleading, in that $\vhigh$ could be in the negative quadrant).
% \mitodo{I don't understand the intuition here. Why is it not trivial in such a case? If it is not deep-in-the-cone then it is close to one of the extremal vectors. It has to be close to $\vlow$ since it is the only one in the positive quadrant, where box-reachable point are located.}
We then notice that all vectors in $\cV$ are either parallel to $\vlow$, or  increase the inner product with $\fhigh$ by at least 1 (due to the integer coordinates). 
Thus, if $\vec{v}-2\spos$ is far enough in the quadrant but close to $\vlow$, reaching it requires a lot of vectors parallel to $\vlow$ (otherwise the path strays too far from $\vlow$, and cannot come back). We can use these vectors as a ``box-safe padding'', and use a Steinitz path in the middle, to obtain an overall box-reaching path. More precisely, we first advance along $\vlow$ until we are far enough above the $\xc$-axis, and far enough below $\vec{v}_\yc$. We then use a Steinitz path to simulate most of the original path to $\vec{v}$ apart from some more vectors parallel to $\vlow$, while staying well below $\vec{v}_y$. We then complete the path by staying parallel to $\vlow$ until reaching $\vec{v}$.
\hfill \qed

\section{On the Semilinearity of Box Reachability}
\label{sec:semilinearity}
\subsection{From Box Reachability in $\dVAS{d}$ to Reachability in $\dVAS{2d}$}
Having established the eventual-equivalence of $\boxreach{\cV}$ and $\reach{\cV}$ for $\dVAS{2}$ in $\clshape$, it is natural to ask what happens in $\lshape$, where the two sets do not coincide. In the following, we show that $\boxreach{\cV}$ is semilinear and in particular is semilinear in the restriction to $\lshape$. In fact, we show a more general result, linking box-reachability to standard reachability, at the cost of doubling the dimension (see \cref{apx:thm:d box reach to 2d reach}).
\begin{theorem}
\label{thm:d box reach to 2d reach}
Consider a $\dVAS{d}$ $\cV$, and define the $\dVAS{2d}$ $\cV'$ as 
\[
\cV' = \{(x_1,\ldots,x_d,-x_1,\ldots,-x_d) \mid (x_1,\ldots,x_d) \in \cV\} \cup \{\vec{e_{d+i}} \mid 1 \leq i \leq d\}
\]
where $\vec{e_{d+i}}$ is the unit vector with $1$ at coordinate $d+i$. Then
\[
\boxreach{\cV} = \{(v_1,\dots,v_d) \in \mathbb{N}^d \mid (v_1,\dots,v_d,0,\dots,0) \in \reach{\cV'}\}.
\]
\end{theorem}
% \begin{proof}[Sketch]
% Intuitively, the equivalence is obtained by noticing that in a box reaching path  we can negate all the vectors (but keep their order), provided we start from $\vec{v}$, without violating the VAS semantics. See the proof in .
% \end{proof}
By \cite{hopcroft1979reachability}, the reachability set of $\dVAS{4}$ is semilinear, and in particular the box-reachability set is also semilinear as a projection. By \cref{thm:d box reach to 2d reach} we therefore have:
%the following.
\begin{corollary}
    \label{cor:box reach 2 VAS is semilinear}
    For a $\dVAS{2}$  $\cV$, the set $\boxreach{\cV}$ is semilinear.
\end{corollary}

\subsection{Box Reachability in \oVASS is Semilinear}
\label{sec:box reach in 1 VASS}
As demonstrated in \cref{xmp:2vass and 3vas}, \cref{thm:main:reach coincides with box reach from W} cannot be extended to either $\oVASS$ or $\dVAS{3}$. Still, it is desirable to reason about box-reachability in VASS. 
In this section, we show that the box-reachable set in \oVASS is semilinear.

The precise definitions and proofs are in~\cref{apx:box reach 1 VASS semilinear}. We sketch the main idea.
The fundamental tool we rely on is that for $\oVASS$, reachability can be characterized by \emph{Linear Path Schemes with One Cycle} (1-LPS). These are path schemes of the form $\alpha\beta^*\gamma$, where $\alpha,\gamma$ are paths, and $\beta$ is a cycle. We observe that crucially, while paths of this form might not be box-reaching, they cannot exceed their box by too much. Indeed, they can go outside their box only as much as $\gamma$ or $\beta$ do.

We therefore define the notion of a \emph{closing suffix} -- a (box reaching) path that can be concatenated after $\alpha\beta^*\gamma$ to make it box reaching (see \cref{fig:ell indices sketch}), and prove that we can characterize box reachability using 1-LPS with closing suffixes. We remark that this technique fails when LPS with more than one cycle are needed, e.g., for $\dVASS{2}$.
\begin{figure}[ht]
        \centering
        \begin{tikzpicture}[xscale=0.4,yscale=0.2]
     \draw[
         thick, blue,
         postaction={decorate}
     ]
     (0,4) -- (5,5) --  (2,3) -- (7,4) -- (4,2) -- (9,3) -- (6,1) -- (11,2) -- (8,0);

    \draw[fill,blue] (0,4) circle [radius=0.2] node[above] {};
    
    % Draw path
     \draw[
         thick,
         rounded corners=5pt,
         decoration={
             markings,
             mark=between positions 0 and 1 step 0.05 with {\arrow[scale=0.7]{>}}
         },
         postaction={decorate}
     ]
     (0,0) -- (22,0) -- (22,1) -- (12,1) -- (12,2) -- (20,2) -- (20,3) -- (17,3) -- (17,4) -- (24,4) -- (24,5) -- (16,5) -- (16,6) -- (24,6);

    % Draw point marks
    \foreach \x/\y/\label in {(0,0)/{},(24,6)/$\ell_1(x^\tar)$, (21,6)/$\ell_2$, (16,6)/$\ell_3$, (15,2)/$\ell_4$, (12,2)/$\ell_5$, (10,0)/$\ell_6$, (8,0)/$\ell_7$} {
        \draw[fill] \x circle [radius=0.2] node[above] {\label};
    }
    \end{tikzpicture}
        \caption{The black path is a box-reaching path to $x^\tar$ (the $\yc$ axis is only for readability). Each $\ell_i$ is the beginning of a possible closing suffix (and in particular these suffixes are box-reaching paths).
        The blue zig-zag is an LPS replacing the prefix up to $\ell_m$. Note that the LPS is not box-reaching itself, but its ``overshoot'' is covered by the long closing suffix.}
        \label{fig:ell indices sketch}
    \end{figure}
We remark that the above also gives an alternative proof of \cref{cor:box reach 2 VAS is semilinear} -- indeed, in $\lshape$ we can cast the $\dVAS{2}$ to a $\oVASS$ by using the state space to capture the bounded coordinates.
 
\section{Discussion and Future Work}
\label{sec:discussion}
In this work we define the notion of box-reachability. 
Intuitively, the set of box-reachable vectors exhibit a ``bounded'' behavior that is more relaxed than a fixed, or even an existentially-quantified bound on some of the coordinates~\cite{demri2013selective}.
%Intuitively, this is a notion of boundedness for VAS that is more relaxed than fixed boundedness and even existential boundedness (i.e., whether all reachable vectors are also reachable beneath some fixed/existentially-quantified bound on some of the coordinates~\cite{demri2013selective}).
Conceptually, this allows us to capture models where one wishes to reach a target without using more resources than the end goal.

We prove that despite the restriction this places on reachability, for $\dVAS{2}$ the set of box-reachable configurations coincides with standard reachability for all configurations beyond some threshold. 
As demonstrated in \cref{xmp:2vass and 3vas}, extending \cref{thm:main:reach coincides with box reach from W} is impossible for $\dVAS{3}$ and $\dVASS{1}$. 

We also show that box-reachability is semilinear for $\dVAS{2}$ and $\oVASS$. We leave open the question of whether this can be extended to $\dVAS{d}$ for $3\le d\le 5$ and $\dVASS{2}$, for which standard reachability is semilinear.

The introduction of box-reachability gives rise to several other natural definitions: one could consider \emph{approximate} box reachability, where the box is allowed to be greater than the target by some fixed (additive or multiplicative) constant. Naive definitions of such extensions turn out to be not very interesting, in the sense that using the Steinitz path, even without our tools in \cref{sec: reach and box reach coincide}, suffices to show that they coincide with reachability (since the Steinitz path already exceeds the target only by a little bit). 
Nonetheless, it is possible that specific modelling needs may require some novel tweaks of box-reachability notions, and we hope this research serves as a basis for reasoning about such extensions. 

%The reader may notice that we do not address complexity issues  in this paper. Indeed, we do not study an explicit decision problem, but rather computing a representation of the box-reachability set. To this end, if $\cV$ is given in unary, our output is polynomial, and is exponential otherwise, similarly to classical representations of the reachability set (c.f.,~\cite{blondin2015}).

%\pagebreak

\subsubsection*{Acknowledgments:}
We thank an anonymous referee for suggesting~\cref{thm:d box reach to 2d reach}, which greatly simplified our previous proof of \cref{cor:box reach 2 VAS is semilinear}.\\
S. Almagor is supported by the ISRAEL SCIENCE FOUNDATION (grant No. 989/22)

\bibliographystyle{splncs04}
\bibliography{ref}

%\pagebreak
\appendix
\section{Proofs}
\subsection{Proof of \cref{lem: drop and peak of steinitz}}
\label{apx: drop and peak of steinitz}
%\begin{proof}
Let $\pi$ be a Steinitz path with $\eff(\pi)=(x,y)$. We prove the bounds for $\xc$, the bounds  for $\yc$ are obtained identically.

We start with $\drop_\xc(\pi)$.
Note that $\effx(\pi[1])\ge -2\norm{\cV}$, since the maximal effect of a single transition is $\norm{\cV}$. 
Consider $2\le i\le |\pi|$, then we can use the properties of $\pi$ as a Steinitz path. 

First, if $\effx(\pi[1,i])\ge \frac{i-2}{|\pi|}x$, then in particular $\effx(\pi[1,i])\ge 0\ge -2\norm{\cV}$ and we are done.
Otherwise, we have 
\[
\frac{i-2}{|\pi|}x- \effx(\pi[1,i])\le 2\norm{\cV} \implies \effx(\pi[1,i])\ge -2\norm{\cV}. 
\]
so we have $\drop_\xc(\pi)\le 2\norm{\cV}$. 

We proceed with $\peak_\xc(\pi)$.
Let $1\le i\le |\pi|$. Note that for $i=1$ we have $\effx(\pi[1,1])\le \norm{\cV}\le x+2\norm{\cV}$.
Thus, we can focus on $i\ge 2$.
If $\effx(\pi[1,i])\le \frac{i-2}{\pi}x$, then $\effx(\pi[1,i])\le x\le x+2\norm{\cV}$ and we are done.
Otherwise, we use the Steinitz property: 
\[\begin{split}
\effx(\pi[1,i])-\frac{i-2}{\pi}x\le 2\norm{\cV} \implies
\effx(\pi[1,i])\le 2\norm{\cV}+\frac{i-2}{\pi}x\le x+ 2\norm{\cV}
\end{split}
\]
so we have $\peak_\xc(\pi)\le x+2\norm{\cV}$. 
\hfill \qed
%\end{proof}

\subsection{The Case of $\dVAS{1}$}
\label{apx: the 1 VAS case}
\begin{theorem}
    \label{thm:1 VAS box reach coincides with reach eventually}
    Consider a $\dVAS{1}$ $\cV=\{a_1,\ldots,a_k\}\subseteq \bbZ$.
    There exists $M_1\in \bbN$ (polynomial in $\norm{\cV}$) such that $\reach{\cV}\cap [M_1,\infty]=\boxreach{\cV}\cap [M_1,\infty]$.  
\end{theorem}
\begin{proof}
We consider the setting of a $\dVAS{1}$ $\cV=\{a_1,\ldots,a_k\}\subseteq \bbZ$.
We show that there exists $M\in \bbN$ (polynomial in $\norm{\cV}$) such that for every $n\ge M$ we have that $n\in \reach{\cV}$ if and only if $0\Bleads{\pi}n$.

Let $d=\min\{a \in \cV\mid a>0\}$. If there are no positive numbers in $\cV$ then $\reach{\cV}=\{0\}$. Thus, we can assume $d>0$.
Intuitively, we claim that if $n$ is large enough, then it can be reached by a short path to some small $k$ with $k\equiv n\bmod d$, followed by enough additions of $d$, and such a path is box-reaching.

Let $M=2\norm{\cV}^3$, and consider a reachable $n>M$ with a path $0\Nleads{\pi}n$. Let $\rho$ be the maximal suffix of $\pi$ that does not go below $\norm{\cV}^3$, then $\eff(\rho)\ge \norm{\cV}^3$ (as $\rho$ starts around $\norm{\cV}^3$ and needs to reach beyond $M$).
%\mitodo{I think that $M$ needs to be changed to be $2\norm{\cV}^3 +\norm{\cV}$ because the prefix $\rho$ can get us to a counter value of at most $\norm{\cV}^3 + \norm{\cV} - 1$}
Since each transition adds at most $\norm{\cV}$, we have $|\rho|>\norm{\cV}^2$. 
In particular, there is some positive $b\in \cV$ such that $b$ appears at least $\norm{\cV}>d$ times as a transition in $\rho$.
We can remove $d$ occurrences of $b$ to obtain a shorter path that reaches a value with the same modulo $d$ as $n$. 

By repeating this argument, we have that for $n>M$ it holds that $n$ is reachable if and only if there exists $k\le \norm{\cV}^3$ such that $k$ is reachable and $k\equiv n\bmod d$.
%\mitodo{Why it holds only if $k\le \norm{\cV}^3$? Where does this bound come from?}

For each reachable $k\le \norm{\cV}^3$, we can find a path $\rho_k$ such that $\vec{0}\Nleads{\rho_k}k$ and moreover, $\peak(\rho_k)$ is polynomial in $\norm{\cV}$ (by \cref{lem:blondin one VASS}, although a direct proof in this naive case is easy).
%\mitodo{Wondering, what does a direct proof of this look like?}

Let $M_1$ be the maximum among these $\peak(\rho_k)$, then for every $n\ge M_1$, if $n\in \reach{\cV}$ then there exists a reachable $k<\norm{\cV}^3$ such that $k\equiv n\bmod d$. We then box-reach $n$ by reaching $k$ via a path $\rho_k$ that stays lower than $M_1$, and then taking $d$ repeatedly until reaching $n$. Thus, $n\in\boxreach{\cV}$.
\hfill \qed
\end{proof}

\subsection{Proof of First Case -- The Cone Contains the Positive Quadrant}
\label{apx:proof of first case}
Set $W=16\norm{\cV}^3+M$, and note that indeed $W$ is polynomial in $\norm{\cV}$
since $M$ is polynomial in $\norm{\cV}$ (see remark after \cref{lem:ditc}).
Consider $\vec{v} \in \clshape\cap \reach{\cV}$. Our goal is to show that $\vec{v}\in \boxreach{\cV}$. 
Since $\vec{v}\in \clshape$, then $\vec{v} \ge (W,W)$. 
By \cref{obs:spos norm} we have $\norm{\spos} \leq 8\norm{\cV}^3$ and $W = 16\norm{\cV}^3 + M$, it follows that $\vec{v} - 2\spos \ge (M,M)$.

We claim that $\vec{v} - 2\spos$ is $M$-deep-in-the-cone. 
Indeed, Consider the H-representation $\cF=\{\vec{f_1},\vec{f_2}\}$ of $\cone(\cV)$, then 
we need to show that $\dotp{\vec{v}-2\spos,\vec{f_i}}\ge M$ for $i\in \{1,2\}$.

As discussed above, since $\cone(\cV)$ contains $\posquad$, then each extremal vector of $\cone(\vec{\chi_1},\vec{\chi_2})$ is of the form $(\ge 0,\le 0)$ or $(\le 0,\ge 0)$. Specifically, we can write 
$\vec{\chi_1}=(x_1,y_1)$ with $x_1\ge 0$ and $y_1\le 0$. Then, we have $\vec{f_1}=(-y_1,x_1)$, and in particular $\vec{f_1}\ge \vec{0}$. Since $f_1\neq \vec{0}$ has integer coordinates, we have $-y_1+x_1\ge 1$. By identical reasoning, writing $\vec{\chi}_2=(x_2,y_2)$ with $x_2\le 0$ and $y_2\ge 0$, we have $\vec{f_2}=(y_2,-x_2)$ and $y_2-x_2\ge 1$.

Since $\vec{v} - 2\spos \ge (M,M)$ we can write $\vec{v} - 2\spos=(z_1,z_2)$ with $z_1,z_2\ge M$. We then have
\[
\dotp{\vec{v}-2\spos,\vec{f_1}}=z_1\cdot (-y_1)+z_2\cdot x_1\ge M(-y_1+x_1)\ge M
\]
and similarly $\dotp{\vec{v}-2\spos,\vec{f_2}}\ge M$, concluding that $\vec{v}-2\spos$ is $M$-deep in the cone.

%Since $\cone(\cV)$ contains $\posquad$, then each extremal vector of the cone is of the form $(>0,\le 0)$ or $(\le 0,>0)$, as depicted by the red edges in \cref{fig:first case cone contains pos quad} (unless $\cone(\cV)=\bbR^2$, in which case there are no faces and the claim holds trivially).

%Thus, if $\vec{f}=(f_1,f_2)$ is orthogonal to an extremal vector, then either  

%As each $\vec{f}\in \cF$ is a normal to an extremal vector of the cone, it follows that $f_1,f_2\in\bbN$, and at least one of them is positive.
% \itodo{Is it enough? or we should talk about the angle between the facets and why it cannot be more than 180 degrees?}
% \shtodo{If there's some more formal way to discuss the properties of $f_1,f_2$, it would be more convincing. Also it would help with my comment below:}
% Thus, if $\vec{v}-2\spos=(x,y)>(M,M)$, then $\tup{\vec{f},\vec{v-2\spos}}=f_1 x+f_2 y\ge f_1 M + f_2 M = (f_1+f_2)M\ge M$, 
% \shtodo{Why is $f_1+f_2\ge 1$? The fact that at least one of them is positive doesn't readily show this.}
% \itodo{They are both natural, at least one of them is not 0. I think that a more formal explanation about this is necessary, I will try to add one}
% \shtodo{Oh, I see now. Ok, we can leave it like that for now and think about it again later. }
%that is,  $\vec{v} - 2\spos$ is deep-in-the-cone. 

As our target $\vec{v}\in \reach{\cV}$, then in particular $\vec{v}\in \intCone(\cV)$, i.e., there are $\alpha_1,\dots,\alpha_k\in\bbN$ such that $\vec{v}=\alpha_1\vec{v_1}+\dots+\alpha_k\vec{v_k}$. 
Then we can write $\vec{v}-2\spos=\alpha_1\vec{v_1}+\ldots+\alpha_k\vec{v_k}-2\spos$. Recall that $\spos$ is also an integer combination of vectors from $\cV$,
%\mitodo{Why can it have negative coefficients? In the construction of the path to it there are no negative ones.}
%Since $\vec{v}\in \reach{\cV}$, there exist $\alpha_1,\dots,\alpha_k\in\bbN$ such that $\vec{v}=\alpha_1v_1+\dots+\alpha_kv_k$. Thus, $\vec{v}-2\spos=\alpha_1v_1+\dots\alpha_kv_k-2\spos$. Recall that $\spos$ is a a linear combination of (at most two) vectors from $\cV$, 
% \shtodo{How is this at most two vectors? \\
% What does at most two vectors got to do with the lattice?}
% \itodo{I meant that $\spos$ is defined by one vector or from two vectors, according to the signs of the coordinates. In both cases, it is in the lattice as a linear combination of vectors with integer coefficients.}
so in particular $\vec{v}-2\spos\in\lattice(\cV)$. 
We now invoke the powerful Deep-in-the-Cone argument (\cref{lem:ditc}): since 
$\vec{v} - 2\spos$ is both $M$-deep-in-the-cone and in $\lattice(\cV)$, then $\vec{v}-2\spos\in\intCone(\cV)$. 
That is, $\vec{0}\Zleads{\rho}\vec{v}-2\spos$ for some $\bbZ$-path $\rho$. 
%there is a $\bbZ$-path $\rho$ from $(0,0)$ to $\vec{v}-2\spos$. 
However, $\rho$ is not necessarily an $\bbN$-path in $\cV$, as the counters may become negative along it. 

Our next powerful tool is the Steinitz path. We apply \cref{lem:steinitz} to obtain from $\rho$ a $\bbZ$-path $\rho'$ that, intuitively, does not stray too far from the line connecting $\vec{0}$ to $\vec{v}-2\spos$ (see \cref{fig:first case cone contains pos quad}). 
More precisely, for every $2\le i\le |\rho'|$ we have
%$\rho'=\vec{u_1},\ldots,\vec{u_m}$ such that for every $2\le i\le m$ we have 
$\norm{\eff(\rho'[1,i])-\frac{i-2}{|\rho'|} (\vec{v}-2\spos)}\le 2\norm{\cV}$.

Recall (from \cref{obs:spos norm}) that $\spos$ is box-reachable in $\cV$, and let $\theta$ be a path such that $\vec{0}\Bleads{\theta}\spos$.
We now consider the path $\pi=\theta,\rho',\theta$. 
We claim that $\vec{0}\Bleads{\pi}\vec{v}$. Indeed, first note that $\eff(\pi)=2\eff(\theta)+\eff(\rho')=2\spos+\vec{v}-2\spos=\vec{v}$. It remains to show that $\pi$ is non-negative and in the $\vec{v}$-box.

%\paragraph{$\pi$ is non-negative.} 
\textbf{$\pi$ is non-negative.\quad} Since $\vec{0}\Bleads{\theta}\spos$, then up to the prefix $\theta$, the path $\pi$ remains non-negative. Moreover, from $\theta,\rho'$ (i.e. after reaching $\vec{v}-\spos$), the suffix is $\theta$, which remains non-negative even from $\vec{0}$, let alone from $\vec{v}-\spos$. It remains to check the infix $\rho'$.
Recall from \cref{obs:spos norm} that $\spos\ge (2\norm{\cV},2\norm{\cV})$. It is therefore enough to show that $\drop_\xc(\rho')\le 2\norm{\cV}$ and $\drop_\yc(\rho')\le 2\norm{\cV}$.
Since $\eff(\rho')=\vec{v}-2\spos\ge (M,M)$, then this is precisely guaranteed by \cref{lem: drop and peak of steinitz}.

%\paragraph{$\pi$ remains in the $\vec{v}$ box}
\textbf{$\pi$ remains in the $\vec{v}$ box.\quad}
Since $\spos$ is box-reachable via $\theta$, then the $\theta$ prefix of $\pi$ cannot violate box-reachability. Similarly, the $\theta$ suffix cannot cause a box-reachability violation. 
Therefore, the only possible violation is if either $\peak_\xc(\theta\cdot \rho')>\vec{v}_\xc$ or $\peak_\yc(\theta\cdot \rho')>\vec{v}_\yc$.
%$\effx(\theta\cdot \rho'[1,i])>\vec{v}_\xc$ or $\effy(\theta\cdot \rho'[1,i])>\vec{v}_\yc$  for some $i\le |\rho'|$, 
where $\vec{v}=(\vec{v}_\xc,\vec{v}_\yc)$. Equivalently, if  $\peak_\xc(\rho')>\vec{v}_\xc-{\spos}_\xc$ or $\peak_\yc(\rho')>\vec{v}_\yc-{\spos}_\yc$.
However, again by \cref{lem: drop and peak of steinitz} we have $\peak_\xc(\rho')\le (\vec{v}_\xc-2{\spos}_\xc)+2\norm{\cV}\le \vec{v}_\xc$ and similarly $\peak_\yc(\rho')\le (\vec{v}_\yc-2{\spos}_\yc)+2\norm{\cV}\le \vec{v}_\yc$,
so the entire path is box-reaching.

We conclude that in this case $\vec{v}\in \boxreach{\cV}$, so we are done.

\subsection{Proof of Second Case -- The Cone Does not contain the Positive Quadrant}
\label{apx:proof of second case}
Set $W=16\norm{\cV}^4+4\norm{\cV}+\norm{\cV}M$.
If $\vec{v}-2\spos$ is not $M$-deep in the cone, 
there exists $\vec{f}\in \cF$ such that $\tup{\vec{f},\vec{v}}<M$.
Intuitively, it may seem necessary that $\vec{f}=\flow$ (since $\vlow$ is ``closer'' to all vectors in the positive quadrant). 
However, this is not necessarily the case: it could be that $\vhigh$ is almost $180^\circ$ from $\vlow$ (in which case \cref{fig:cone contains y axis} is misleading), but has smaller magnitude, and thus has smaller inner products. E.g., if $\vlow=(100,100)$ and $\vhigh=(-2,-1)$.

However, we claim that even in this case, $\dotp{\flow,\vec{v}-2\spos}$ is not too large. Indeed, if $\vhigh$ is in the second quadrant, i.e., of the form $(<0,\ge 0)$ then by similar considerations to the first case we have $\dotp{\fhigh,\vec{v}-2\spos}>M$, so $\vec{f}=\flow$ and we are done.

Otherwise, $\fhigh$ is of the form $(<0,<0)$. In \cref{apx:close is close to vlow} we show that in this case, if  $\dotp{\fhigh,\vec{v}-2\spos}\le M$ then $\dotp{\flow,\vec{v}-2\spos}\le \norm{\cV}M$.
Thus, in both cases we can assume $\dotp{\flow,\vec{v}-2\spos}\le \norm{\cV}M$, which we do henceforth.

By \cref{obs:spos norm} and since $\norm{\flow}\le \norm{\cV}$ we have:
\[\dotp{\vec{v},\flow}\le \norm{\cV}M+2\dotp{\spos,\flow}\le \norm{\cV}M+2\norm{\spos}\norm{\flow}\le \norm{\cV}M+16\norm{\cV}^3\norm{\cV}=\norm{\cV}M+16\norm{\cV}^4\]

Consider some vector $\vec{u}\in \cV$.
We claim that $\dotp{\vec{u},\flow}\ge 0$. Indeed, otherwise $\vec{u}$ points at the negative halfspace defined by $\flow$, which is a contradiction to $\vlow$ being the extremal vector of $\cone(\cV)$. Since all our vectors are integers, it follows that for every $\vec{u}\in \cV$, either $\dotp{\vec{u},\flow}=0$ or $\dotp{\vec{u},\flow}\ge 1$.
%\mitodo{Why is it a contradiction to $\flow$ being the extremal vector of $\cone(\cV)$?}

Since $\vec{v}\in \reach{\cV}$, there is a path $\vec{0}\Nleads{\pi}\vec{v}$ with $\pi=\vec{u_1}\vec{u_2}\cdots \vec{u_k}$. By the above, each step $\vec{u_i}$ of $\pi$ either maintains the inner product with $\flow$, or increases it by at least 1, but the sum of these products is $\dotp{\vec{v},\flow}$. By the upper bound above, we have $\dotp{\eff(\pi),\flow}=\sum_{i=1}^k\dotp{\vec{u_i},\flow}\le \norm{\cV}M+16\norm{\cV}^4$. 
In particular, we have $|\{i\mid \dotp{\vec{u_i},\flow}>0\}|\le \norm{\cV}M+16\norm{\cV}^4$.

Since $\vec{v}\ge (W,W)$, and the maximal effect of each transition is $\norm{\cV}$, we have $|\pi|\ge W/\norm{\cV}\ge 16\norm{\cV}^3$. Moreover, the number of vectors $\vec{u_i}$ with $\effx(\vec{u_i})>0$ is at least $W/\norm{\cV}\ge 16\norm{\cV}^3$. In the following we refer to this as $|\pi|$ for brevity (instead of defining a cumbersome increasing subsequence of indices).

Thus, in order for the inner products to remain low enough, we must have $|\{i\mid \dotp{\vec{u_i},\flow}=0\wedge \effx(\vec{u_i})>0\}|\ge 16\norm{\cV}^4+4\norm{\cV}+\norm{\cV}M - (\norm{\cV}M+16\norm{\cV}^4)= 4\norm{\cV}$.
%\mitodo{Why is this inequality correct? There are at least $W/\norm{\cV}$ vectors with $eff_x > 0$ and $|\{i\mid \dotp{\vec{u_i},\flow}>0\}|\le \norm{\cV}M+16\norm{\cV}^4$. So, it looks like the lower bound here is $|\{i\mid \dotp{\vec{u_i},\flow}=0\wedge \effx(\vec{u_i})>0\}|\ge (16\norm{\cV}^4+4\norm{\cV}+\norm{\cV}M)//\norm{\cV} - (\norm{\cV}M+16\norm{\cV}^4)$ which can be negative and to obtain the bound that is described originally, we need to multiply $W$ by a factor of $\cV$ (need to address the number of steps with positive effect and not the total effect)}
That is, there are many vectors in $\pi$ that are parallel to the lower facet $\vlow$ of the cone.
Intuitively, we now construct a box-reaching path from $\vec{0}$ to $\vec{v}$ as follows: we start by staying parallel to $\vlow$ until we reach beyond $(2\norm{\cV},2\norm{\cV})$ (which is the Steinitz constant), we then take a Steinitz path (\cref{lem:steinitz}) to reach a vector whose inner product with $\flow$ is the same as that of $\vec{v}$. Finally, we traverse the remainder of the path using more parallel vectors to $\vlow$.

By guaranteeing that the effect of the prefix and suffix are large enough, we can ensure the Steinitz path does not exceed  the $\vec{v}$ box using \cref{lem: drop and peak of steinitz}.

We give the precise details of the latter part. Let $\vec{p_1},\ldots \vec{p_{2\norm{\cV}}},\vec{q_1},\ldots \vec{q_{2\norm{\cV}}}$ be $4\norm{\cV}$ vectors occurring in distinct indices of $\pi$ such that all of them are orthogonal to $\flow$ and have positive $\effx$. Construct paths $\mu=\vec{p_1},\ldots \vec{p_{2\norm{\cV}}}$ and $\eta=\vec{q_1},\ldots \vec{q_{2\norm{\cV}}}$, then since all the vectors are parallel to $\vlow$ and have positive $\effx$, then all the vectors are positive scalar products of $\vlow$ and in particular also have positive $\effy$. It follows that both $\mu$ and $\eta$ are box-reaching, and that $\eff(\mu),\eff(\eta)\ge (2\norm{\cV},2\norm{\cV})$. 

Let $\xi$ be a path constructed by the remaining vectors of $\pi$, after removing $\mu$ and $\eta$. Since $\eff(\mu\xi\eta)=\eff(\pi)$, we have $\vec{0}\Zleads{\mu\xi\eta}\vec{v}$. 
Let $\vec{p}=\eff(\mu)$ and $\vec{q}=\eff(\eta)$, then $\vec{p}\Zleads{\xi}\vec{v}-\vec{q}$.
%\mitodo{Why does $\vec{p}\Zleads{\xi}\vec{q}$? $eff(\xi)=v-p-q$ so taking $\xi$ from $p$ leads to $p+v-p-q=v-q$}
We now apply \cref{lem:steinitz} to obtain a Steinitz path $\xi'$ such that $\vec{p}\Zleads{\xi'}\vec{q}$. In particular we have $\eff(\xi')=\vec{v}-\vec{q}-\vec{p}$.
% and for every $i\ge 2$ we have 
% \[
% \norm{\eff(\xi')-\frac{i-2}{|\xi'|}(\vec{q}-\vec{p})}\le 2\norm{\cV}
% \]
We claim that $\vec{0}\Bleads{\mu\xi'\eta}\vec{v}$. Clearly $\vec{0}\Zleads{\mu\xi'\eta}\vec{v}$, it remains to show that the path remains positive and within the $\vec{v}$ box.

%\paragraph{$\mu\xi'\eta$ is non-negative.} 
\textbf{$\mu\xi'\eta$ is non-negative.\quad}
Since $\vec{0}\Bleads{\mu}\vec{p}$, then up to the prefix $\mu$, the path $\mu\xi'\eta$ remains non-negative. Moreover, from $\mu\xi'$ (i.e. after reaching $\vec{v}-\vec{q}$), the suffix is $\eta$, which remains non-negative even from $\vec{0}$, let alone from $\vec{v}-\vec{q}$. It remains to check the infix $\xi'$.
Recall that $\eff(\mu)\ge (2\norm{\cV},2\norm{\cV})$. It is therefore enough to show that $\drop_\xc(\xi')\le 2\norm{\cV}$ and $\drop_\yc(\xi')\le 2\norm{\cV}$. This would readily follow from \cref{lem: drop and peak of steinitz}, provided $\eff(\xi')\ge \vec{0}$. 

To show the latter, observe that $\eff(\mu)\le |\mu|(\norm{\cV},\norm{\cV})=(2\norm{\cV}^2,2\norm{\cV}^2)$, and similarly $\eff(\eta)\le (2\norm{\cV}^2,2\norm{\cV}^2)$. It therefore follows that $\eff(\xi)\ge \eff(\pi)-2(2\norm{\cV}^2,2\norm{\cV}^2)\ge \vec{0}$ (since $\eff(\pi)\ge (W,W)$ and $W\ge 16\norm{\cV}^3$).

\textbf{$\pi$ remains in the $\vec{v}$ box. \quad}
Since $\mu$ and $\eta$ are box-reaching, the only possible violation is if 
$\peak_\xc(\xi')>(\vec{v}-\vec{p})_\xc\ge \vec{v}_\xc-2\norm{\cV}$ or
$\peak_\yc(\xi')>(\vec{v}-\vec{p})_\yc\ge \vec{v}_\yc-2\norm{\cV}$.
This, however, is not the case by \cref{lem: drop and peak of steinitz}.

We conclude that $\vec{0}\Bleads{\mu\xi'\eta}\vec{v}$, so $\vec{v}\in \boxreach{\cV}$ and we are done.

\subsection{Vectors in $\clshape$ are Close to $\vlow$.}
\label{apx:close is close to vlow}
\begin{proposition}
\label{prop:fhigh is far away}
    Let $C=\cone(\{\vec{v_1},\vec{v_2}\})$ where $\vec{v_1}=(x_1,y_1)$ is $(>0,>0)$ and $\vec{v_2}=(-x_2,y_2)$ is $(<0,<0)$ (and the angle between them is less than 180 counter-clockwise). If the normals of $\vec{v_1}$ and $\vec{v_2}$ that define the cone are $\vec{f_1}$ and $\vec{f_2}$, then $\dotp{\vec{v},\vec{f_1}}\le \norm{\cV}\dotp{\vec{v},\vec{f_2}}$.
    %, if $v=(x,y)>(M,M)$ is in $C$, then $\tup{v,f_2}\ge M$.
\end{proposition}
%\mitodo{Why can we assume this about the angle being at most 180?}
\begin{proof}
    Due to the angle between the vertices, their respective normals defining the cone are $\vec{f_1}=(-y_1,x_1)$ and $\vec{f_2}=(y_2,x_2)$.
    It holds that $\tup{\vec{v},\vec{f_2}}=xy_2 + yx_2$.
    % We split into cases:
    % \begin{itemize}
    %     \item if $y_2\ge 0$, then $xy_2\ge 0$. Notice that $yx_2\ge M$ as $y>M$ and $x_2\in\bbN^+$ (we cannot have $x_2=0$ as the angle would be more than $180^\circ$).
    %     \shtodo{Actually, we can have $x_2=0$ and $y>0$, no?}
    %     Thus, $xy_2+yx_2\ge M$.

    We examine the angles of the vectors: the angles that $\vec{v_1}$ and $\vec{v_2}$ induce are $\alpha=\arctan(\frac{y_1}{x_1})$ and $\beta=\arctan(\frac{-y_2}{x_2})+\pi$, respectively.
    We know that $\alpha<\beta<\alpha+\pi$. The latter gives us $\arctan(\frac{-y_2}{x_2})+\pi < \arctan(\frac{y_1}{x_1}) + \pi$. By monotonicity, $\frac{-y_2}{x_2}<\frac{y_1}{x_1}$. 
        Therefore, $y_2>-\frac{y_1}{x_1}x_2$.
        %\shtodo{What happened to the minus sign? $y_2>-\frac{y_1}{x_1}x_2$}
        We get that 
        \[\dotp{\vec{v},\vec{f_2}}=xy_2+yx_2>-x\frac{y_1}{x_1}x_2+yx_2=\frac{x_2}{x_1}(yx_1-xy_1)=\frac{x_2}{x_1}\dotp{\vec{v},\vec{f_1}}\ge \frac{1}{\norm{\cV}}\dotp{\vec{v},\vec{f_1}}\] 
        where the last inequality follows since $1\le x_1,x_2\le \norm{\cV}$, so the ratio is bounded.

        We now have $\dotp{\vec{v},\vec{f_1}}\le \norm{\cV}\dotp{\vec{v},\vec{f_2}}$, so we are done.
    %\end{itemize}
    \hfill \qed
\end{proof}

\subsection{Proof of \cref{thm:d box reach to 2d reach}}
\label{apx:thm:d box reach to 2d reach}
Denote 
\[
R = \{(v_1,\dots,v_d) \in \mathbb{N}^d \mid (v_1,\dots,v_d,0,\dots,0) \in \reach{\cV'}\}.
\]
Then we need to prove that $\boxreach{\cV} = R$.

For the first inclusion, take $\vec{v} = (v_1,\dots,v_d) \in \boxreach{\cV}$ via path $\pi$ in $\cV$ with $\vec{0} \Bleads{\pi} \vec{v}$. Construct $\pi'$ in $\cV'$ by first taking $v_i$ copies of $e_{d+i}$ for each $1 \leq i \leq d$, reaching  $(0,\dots,0,v_1,\dots,v_d)$. Then apply each $\vec{w_i} = (w_i^{(1)},\dots,w_i^{(d)})$ in $\pi$ as $(w_i^{(1)},\dots,w_i^{(d)},-w_i^{(1)},\dots,-w_i^{(d)})$. 

After processing $k$ vectors of $\pi$, the configuration is
\[
\left( \sum_{i=1}^k w_i^{(1)}, \dots, \sum_{i=1}^k w_i^{(d)}, v_1 - \sum_{i=1}^k w_i^{(1)}, \dots, v_d - \sum_{i=1}^k w_i^{(d)} \right)
\]
Since $\pi$ is box-reaching, each partial sum satisfies $0 \leq \sum_{i=1}^k w_i^{(j)} \leq v_j$ for $1 \leq j \leq d$, ensuring all coordinates remain non-negative. The final state is $(v_1,\dots,v_d,0,\dots,0)$, proving $\vec{v} \in R$.

For the converse inclusion, take $\vec{v} = \left(v_1,\dots,v_d\right) \in R$ via path $\pi$ in $\cV'$ with $\vec{0}\Nleads{\pi}\left(v_1,\dots,v_d,0,\dots,0\right)$. 
We first notice that we can safely rearrange $\pi$ to apply all $e_{d+j}$ vectors first, reaching $\left(0,\dots,0,\alpha_1,\dots,\alpha_d\right)$, then apply all vectors of the form $\left(w_i^{(1)},\dots,w_i^{(d)},-w_i^{(1)},\dots,-w_i^{(d)}\right)$. 
Indeed, since $e_{d+j}$ has only non-negative coordinates, it can be taken earlier without violating the VAS semantics (i.e., without causing a counter to become negative).

The final configuration therefore implies for each $1 \leq j \leq d$ that:
\begin{align*}
\sum_{i} w_i^{(j)} &= v_j \\
\alpha_j &= v_j 
\end{align*}

Now consider the projection $\rho$ of $\pi$ to the first $d$ coordinates gives a $\cV$-path to $\vec{v}$. After $m$ steps, the configuration of $\rho$ is
\[
\left( \sum_{i=1}^m w_i^{(1)}, \dots, \sum_{i=1}^m w_i^{(d)}, v^{(1)} - \sum_{i=1}^m w_i^{(1)}, \dots, v^{(d)} - \sum_{i=1}^m w_i^{(d)} \right)
\]
Non-negativity of coordinates $d+1$ to $2d$ in $\pi$ implies $\sum_{i=1}^m w_i^{(j)} \leq v_j$ for each $j$, so $\rho$ stays within the box induced by $\vec{v}$, proving $\vec{v} \in \boxreach{\cV}$.

\subsection{Box Reachability in \oVASS is Semilinear}
\label{apx:box reach 1 VASS semilinear}
Formally, a \oVASS is $\cA=\tup{Q,T}$ where $Q$ is a finite set of states, and $T\subseteq Q\times \bbZ\times Q$ is a transition relation. A \emph{configuration} of $\cA$ is $(x,q)\in \bbZ\times Q$ where $x$ is the counter value and $q$ is a state. A path is a sequence of configurations $\pi=(x_1,q_1),\ldots,(x_n,q_n)$ where $((x_i,q_i),(x_{i+1}-x_i),(x_{i+1},q_{i+1}))\in T$ for every $1\le i<n$. An $\bbN$-path is a path where the counter value remains in $\bbN$. A configuration $(x,q)$ is \emph{reachable} from $(x',q')$ if there is an $\bbN$-path from $(x',q')$ to $(x,q)$.

We can fit our definition of box reachability to $\dVASS{1}$ as follows: we say that a configuration $(x^\tar,q^\tar)$ is \emph{box-reachable from $q_0\in Q$} if there is a path $(0,q_0)\Nleads{\pi}(x^\tar,q^\tar)$ such that $\eff(\pi[1,i])\le x^\tar$ for all $i\le |\pi|$. Note that the states do not restrict box-reachability.
For a $\dVASS{1}$ $\tup{Q,T}$ and state $q_0\in Q$, we then denote by $\boxreach{\tup{Q,T},q_0\to q^\tar}$ the set of values $x^\tar\in \bbN$ such that $(x^\tar,q^\tar)$ is box-reachable from $(0,q_0)$.
We use the definitions of $\drop,\peak,\eff$ by referring only to the counter, disregarding the state.

\begin{example}[Box Reachability in $\oVASS$]
\label{xmp:box reach near axes}
Consider the \dVAS{2} $\cV=\{(1,7),(3,-6)(-2,6)\}$ and the target vector \(\vec{t} = (6,8)\). It is box-reachable via the path depicted in \cref{fig:zig zag path near axes}.

We view this $\dVAS{2}$ as a \oVASS by restricting the $\yc$ coordinate to be at most $8$, and treating $\{0,\ldots,8\}$ as the state space, and $\xc$ as the 1-dimensional vectors.
Thus, state $8$ with vector $6$ is box-reachable. 

Observe that the initial \((1,7)\) enables the $(3,-6),(-2,6)$ cycle to be taken initially (and repeatedly). The final $(1,7)$ ``closes'' this path so that it is within its box. 

\begin{figure}[ht]
    \centering
\begin{tikzpicture}[scale=0.5]
    % Axes
    \draw[->] (0,0) -- (24,0) node[right] {$x$};
    \draw[->] (0,0) -- (0,4) node[above] {$y$};

    % Points and path
    \draw[fill] (0,0) circle [radius=0.1] node[above left] {$(0,0)$};
    \draw[fill] (4,3.5) circle [radius=0.1] node[right] {$(1,7)$};
    \draw[fill] (13,0.5) circle [radius=0.1] node[right] {$(4,1)$};
    \draw[fill] (8,3.5) circle [radius=0.1] node[right] {$(2,7)$};
    \draw[fill] (17,0.5) circle [radius=0.1] node[right] {$(5,1)$};
    \draw[fill] (21,4) circle [radius=0.1] node[right] {$(6,8)$};

    \draw[->] (0,0) -- (4,3.5) node[midway, sloped, above] {$(1,7)$};
    \draw[->] (4,3.5) -- (13,0.5) node[midway, sloped, below] {$(3,-6)$};
    \draw[->] (13,0.5) -- (8,3.5) node[pos=0.2, sloped, above] {$(-2,6)$};
    \draw[->] (8,3.5) -- (17,0.5) node[midway, sloped, above] {$(3,-6)$};
    \draw[->] (17,0.5) -- (21,4) node[midway, sloped, above] {$(1,7)$};

    % Box
    \draw[dashed] (0,4) -- (21,4);
    \draw[dashed] (21,0) -- (21,4);
\end{tikzpicture}

    \caption{A zig-zag box-reaching path of \cref{xmp:box reach near axes}.}
    \label{fig:zig zag path near axes}
\end{figure}
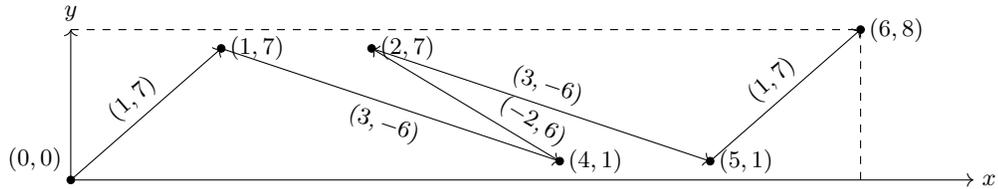
%
%However, note that no strict prefix of this path after the initial $(1,7)$ is box-reaching. Moreover, the configuration $(5,1)$ reached along this path is not box-reachable (indeed, no configuration below $\yc=7$ is box-reachable in $\cV$).
\end{example}

Fix an \oVASS $\tup{Q,T}$ for the remainder of the section.

We first recall a fundamental result about reachability in \oVASS from~\cite{leroux2004flatness,blondin2021reachability}, which intuitively states that reachability can be characterized using \emph{Linear Path Schemes} with one cycle. 
A \emph{Linear Path Scheme}\footnote{Usually LPS can have more than one cycle, but for \oVASS this is unnecessary.} (LPS, for short) is an expression of the form  $\alpha\beta^*\gamma$ where $\alpha,\beta,\gamma$ are paths and $(x,q)\Nleads{\alpha}(x_1,q_1)$, $(0,q_1)\Zleads{\beta}(\Delta_x,q_1)$, $(x'_1,q_1)\Nleads{\gamma}(x',q')$ for some $x_1,q_1,\Delta_x$ and $x'_1$. Note that $\beta$ is a \emph{cycle} on the state $q_1$. A path is \emph{induced} by an LPS $\alpha\beta^*\gamma$ is of the form $\alpha\beta^k\gamma$ for some $k\in \bbN$.

\begin{lemma}[Proposition 4.3 of \cite{blondin2021reachability}]
\label{lem:blondin one VASS}
    Let $\tup{Q,T}$ be a \oVASS and $(x,q),(x',q')$ configurations. Then $(x',q')$ is reachable from $(x,q)$ if and only if there exists an LPS $\alpha \beta^* \gamma$ such that the following hold:
    \begin{itemize}
        \item There exists some $k\in \bbN$ such that $(x,q)\Nleads{\alpha\beta^k\gamma}(x',q')$.
        \item $|\alpha|,|\beta|,|\gamma|\le p_1(|Q|,\norm{T})$ where $p_1$ is a fixed polynomial,\footnote{The proof in \cite{blondin2021reachability} in fact shows $p_1$ is of degree at most $3$.} independent of $\tup{Q,T}$ .
    \end{itemize}
\end{lemma}
For our fixed \oVASS $\tup{Q,T}$, let $\BLPS=p_1(|Q|,\norm{T})$. 
In light of \cref{lem:blondin one VASS}, we define the set of \emph{short LPS} of $\tup{Q,T}$ as $\LPS=\{\alpha\beta^*\gamma\mid |\alpha|,|\beta|,|\gamma|\le \BLPS\}$, and observe that the lemma states that the reachability relation of $\tup{Q,T}$ is in fact captured by paths induced by an LPS in $\LPS$.

In general, paths induced by an LPS are not box-reaching. Moreover, in general it might not be possible to permute the order of transitions of such a path so that it becomes box-reachable. %This can be seen in e.g., \cref{xmp:box reach near axes}.
Despite this, %, and still within the motivation of \cref{xmp:box reach near axes}, 
we show that even when box reachability is concerned, ``most'' of the path can be taken via an LPS, and only some additions may be needed at the end to ``close'' the box. 

Consider a path $\pi$. We define its \emph{overshoot} as $\oshoot(\pi)=\peak(\pi)-\eff(\pi)$. Thus, the overshoot measures by how much $\pi$ violates box-reachability.
Consider an LPS $\alpha\beta^*\gamma$ with $\eff(\beta)>0$ and an induced path $\alpha\beta^k\gamma$ with a very large $k$, specifically such that $k\cdot \eff(\beta)>\peak(\alpha)$.
That is, a path that takes a positive cycle enough times to go well beyond the configurations reached by $\alpha$.
We notice the following.
\begin{itemize}
    \item $\eff(\alpha\beta^k\gamma)=\eff(\alpha)+k\cdot \eff(\beta)+\eff(\gamma)$ (simply a sum of the effects).
    \item $\peak(\alpha\beta^k\gamma)=\max\{\eff(\alpha)+k\cdot \eff(\beta)+\peak(\gamma),\eff(\alpha)+(k-1)\eff(\beta)+\peak(\beta)\}$ where the maximum depends on whether the last iteration of $\beta$ peaks further than $\gamma$.
\end{itemize}

Therefore, we have $\oshoot(\alpha\beta^k\gamma)=\max\{\oshoot(\gamma),\oshoot(\beta)-\eff(\gamma)\}$. In particular, observe that this is independent of $k$. We therefore define this value as the overshoot of the LPS, i.e., $\oshoot(\alpha\beta^*\gamma)=\max\{\oshoot(\gamma),\oshoot(\beta)-\eff(\gamma)\}$. Note that for small values of $k$ (i.e., those where $\peak(\alpha)$ might be beyond $\peak(\beta^k)$), the equality above does not hold, and the value may depend on $\alpha$, but this does not concern us in the following.

\begin{example}
    \label{xmp:zig zag for overshoot}

    Recall the $\dVAS{2}$ of \cref{xmp:box reach near axes}, and view it as a $\dVASS{1}$ where the $\yc$ coordinate is the state. Consider the LPS $\alpha\beta^*\gamma$ with $\alpha=(1,7)$, $\beta=(3,-6),(-2,6)$ and empty $\gamma$. We have $\eff(\alpha\beta^k\gamma)=1+1\cdot k$ (on the $\xc$ axis) and $\peak(\alpha\beta^k\gamma)=1+1\cdot k+2$ (since $\peak(\beta)=3$) and therefore $\oshoot(\alpha\beta^k\gamma)=2$.

    Also observe that the suffix $(3,-6)(1,7)$ ``overcomes'' this $\oshoot$. Below we define it as a \emph{closing} suffix. 
\end{example}

Recall that for every $\alpha\beta^*\gamma\in \LPS$ we have $|\alpha|,|\beta|,|\gamma|\le \BLPS$, and that for every path $\pi$ we have $\eff(\pi)\le \norm{T}\cdot |\pi|$. We can thus \emph{define} $\mshoot=\max\{\oshoot(\alpha\beta^*\gamma)\mid \alpha\beta^*\gamma\in \LPS\}$, and we have  $\mshoot\le p_2(|Q|,\norm{T})$ for some fixed polynomial $p_2$.
%\mitodo{What exactly is the overshoot of an LPS? It was defined for a path $\pi$ and not for an LPS. Does it refer to $over(\alpha\beta\gamma)$? Or does it refer to the maximal $over(\alpha\beta^k\gamma)$? It the latter case, I think the overshoot is unbounded.}
%\shtodo{It was defined two lines above your comment. I now emph'd it.}
%\mitodo{Why is $p_2$ fixed? By the inequality $eff(\pi) \le ||T|| \cdot |\pi|$ it looks like it depends on $||T||$}
%\shtodo{The polynomial is fixed. I.e., it's something like $p_2(x,y)=x^3y^2+15x$. That is, the coefficients and the powers are fixed.}

Let $q_\gamma$ be the last state reached by $\gamma$. Consider a path $\theta$ starting at $q_\gamma$ such that $\drop(\theta)=0$ and $\peak(\theta)=\eff(\theta)$. We say that $\theta$ \emph{closes} the LPS $\alpha\beta^*\gamma$ if $\eff(\theta)\ge \oshoot(\alpha\beta^*\gamma)$.
%\mitodo{Same question as above, what is an overshoot of an LPS?}
%\shtodo{same answer.}
%\mitodo{The case where $eff(\beta) < 0$ isn't considered. What if $\alpha$ moves us to the right, $\beta$ to the left and then $\gamma$ moves us to the right but not enough?}
%\shtodo{The whole part starts with an assumption that $\eff(\beta)>0$, see the paragraph starting with ``Consider a path $\pi$...''}

Notice that if $\theta$ closes $\alpha\beta^*\gamma$, then for large enough $k$ the path $\alpha\beta^k\gamma\theta$ is box-reaching, since $\eff(\theta)$ is large enough for the counter to go above $\peak(\alpha\beta^k\gamma)$, and $\theta$ itself is box-safe, i.e., $\peak(\theta)=\eff(\theta)$.

\begin{remark}
\label{rmk: one cycle is crucial}
We remark that our definition of closing suffix, as well as the rest of the proof, crucially relies on the fact that the schemes in the LPS we consider have a \emph{single} cycle. This prevents us from extending our methods to e.g., $\dVASS{2}$.
\end{remark}

Using the bounds in \cref{lem:blondin one VASS}, we can now characterize box reachability by means of closing paths, as follows.
\begin{lemma}
    \label{lem:box reachability in 1 VASS by closing paths}
    Let $(0,q_0), (x^\tar,q^\tar)$ be configurations such that $x^\tar\ge p_3(|Q|,\norm{T})$ (where $p_3$ is some fixed polynomial).
    Then $(x^\tar,q^\tar)$ is box-reachable from $(0,q_0)$ if and only if there exists  $\alpha\beta^*\gamma\in \LPS$ and a path $\theta$ such that the following hold.
    \begin{itemize}
        \item There exists some $k\in \bbN$ such that $(0,q_0)\Nleads{\alpha\beta^k\gamma\theta}(x^\tar,q^\tar)$.
        \item $\theta$ closes $\alpha\beta^*\gamma$.
        \item $|\theta|\le \norm{T}\cdot \mshoot\cdot |Q|$.
    \end{itemize}
\end{lemma}
\begin{proof}[Sketch]
The complete proof is delegated to \cref{apx:box reachability in 1 VASS by closing paths}.
    Fix $p_3(|Q|,\norm{T})$ much larger than $\BLPS$ and $\mshoot$.
    The ``if'' direction is fairly straightforward. Intuitively, since $x^\tar$ is large enough, then the number of repetitions of $\beta$ (namely $k$) is also large, and therefore the concatenation of $\theta$ actually yields a box-reaching path.

     The converse direction is more involved. We start with a high-level sketch of the proof, depicted in \cref{fig:ell indices}.
    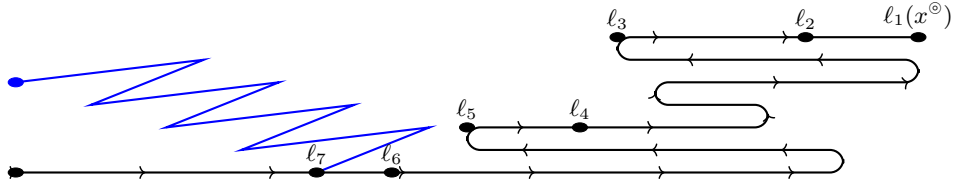
\begin{figure}[ht]
        \centering
        \begin{tikzpicture}[xscale=0.5,yscale=0.3]
     \draw[
         thick, blue,
         postaction={decorate}
     ]
     (0,4) -- (5,5) --  (2,3) -- (7,4) -- (4,2) -- (9,3) -- (6,1) -- (11,2) -- (8,0);

    \draw[fill,blue] (0,4) circle [radius=0.2] node[above] {};
    
    % Draw path
     \draw[
         thick,
         rounded corners=5pt,
         decoration={
             markings,
             mark=between positions 0 and 1 step 0.05 with {\arrow[scale=0.7]{>}}
         },
         postaction={decorate}
     ]
     (0,0) -- (22,0) -- (22,1) -- (12,1) -- (12,2) -- (20,2) -- (20,3) -- (17,3) -- (17,4) -- (24,4) -- (24,5) -- (16,5) -- (16,6) -- (24,6);

    % Draw point marks
    \foreach \x/\y/\label in {(0,0)/{},(24,6)/$\ell_1(x^\tar)$, (21,6)/$\ell_2$, (16,6)/$\ell_3$, (15,2)/$\ell_4$, (12,2)/$\ell_5$, (10,0)/$\ell_6$, (8,0)/$\ell_7$} {
        \draw[fill] \x circle [radius=0.2] node[above] {\label};
    }
    \end{tikzpicture}
        \caption{The sequence of indices defining box-safe suffixes. The $\yc$ axis is only for readability. The black path is a box-reaching path to $x^\tar$, the blue zig-zag is an LPS replacing the prefix up to $\ell_m$. Note that the LPS is not box-reaching itself, but its overshoot is covered by the path suffix}
        \label{fig:ell indices}
    \end{figure}
    Since $(x^\tar,q)$ is box reachable, there exists a box reaching path $(x_0,q_0)\Nleads{\pi}(x^\tar,q^\tar)$. Since $x^\tar$ is very large, then $\pi$ must be long, and in particular traverse many large $\xc$ values.
    
    Starting from $x^\tar$ and working backwards along $\pi$, we identify indices $\ell_m<\ell_{m-1}<\ldots <\ell_1$ such that the suffix $\pi[\ell_i+1,\ldots]$ remains entirely ``to the right'' of $\eff(\pi[1,\ell_i])$. That is, these suffixes are box reaching paths.
    By working backwards enough, we can reach $\ell_m$ such that $\eff(\pi[1,\ell_m])$ is far enough from $x^\tar$ so that the overshoot of any linear path scheme in $\LPS$ is smaller than $x^\tar-\eff(\pi[1,\ell_m])$. 
    This guarantees that we can use a linear path scheme to reach the $\ell_m$ configuration, and continue from there with a box reaching suffix to $(x^\tar,q^\tar)$. Moreover, we can bound the length of this suffix to satisfy the conditions in the lemma.
    \hfill \qed
\end{proof}

We can now characterize the box-reachability set for \oVASS.
\begin{theorem}
    \label{thm:box reach for 1VASS is semilinear}
    For every \oVASS $\tup{Q,T}$ and states $q_0,q^\tar\in Q$, the set $\boxreach{\tup{Q,T},q_0\to q^\tar}$ is effectively semilinear.
\end{theorem}
\begin{proof}
    Consider the finite set $S=\{x^\tar\mid x^\tar\le p_3(|Q|,\norm{T})\wedge \exists \pi.\ (0,q_0)\Bleads{\pi} (x^\tar,q^\tar)\}$ of box-reachable configurations below $p_3(|Q|,\norm{T})$.
    
    By \cref{lem:box reachability in 1 VASS by closing paths} we have that $(x^\tar,y^\tar)$ is box reachable from $(0,q_0)$ if and only if either $x^\tar\in S$, or there exists $\alpha\beta^*\gamma\in \LPS$ and $\theta$ with $|\theta|\le \norm{T}\cdot \mshoot\cdot |Q|$ such that $\theta$ closes $\alpha\beta^*\gamma$ and $(0,q_0)\Nleads{\alpha\beta^k\gamma\theta}(x^\tar,q^\tar)$.

    For every $\alpha\beta^*\gamma\in \LPS$, let $q'$ be the last state in $\gamma$ and define 
    $\Theta(\alpha\beta^*\gamma)=\{\theta\mid (0,q')\Nleads{\theta}(\eff(\theta),q^\tar) \text{ closes } \alpha\beta^*\gamma\}$.
    %\mitodo{I think it should be$(0,q')\Nleads{\theta}(\eff(\theta),q^\tar)$}
    %\shtodo{thanks.}
    %\mitodo{Also, what if this set is empty? Can this even happen?}
    %\shtodo{Could be.  In that case we cannot get to $q^\tar$ in the box.}
    Note that $\Theta(\alpha\beta^*\gamma)$ and $\LPS$ are both finite and effectively computable. 
    Rephrasing the above, we have that $x^\tar\in \boxreach{\tup{Q,T},q_0\to q^\tar}$ if and only if $x^\tar\in S$ or there exists $\alpha\beta^*\gamma\in \LPS$ and $\theta\in \Theta(\alpha\beta^*\gamma)$ such that $x^\tar=\eff(\alpha\beta^k\gamma\theta)$ for some $k>\peak(\alpha)$.
    
    Thus, we have
    \[\boxreach{\tup{Q,T},q_0\to q^\tar}=S\cup \{\eff(\alpha\beta^k\gamma\theta)\mid \alpha\beta^*\gamma\in \LPS,\ \theta\in \Theta(\alpha\beta^*\gamma),k\in \bbN\}\]
    which is a finite union (over $\LPS$ and $\Theta$) of the linear sets $\{\eff(\alpha\beta^k\gamma\theta)\mid k>\peak(\alpha)\}$.
    \hfill \qed
\end{proof}

\subsection{Proof of \cref{lem:box reachability in 1 VASS by closing paths}}
\label{apx:box reachability in 1 VASS by closing paths}
%\begin{proof}
    Fix $p_3(|Q|,\norm{T})=\norm{T}\cdot (\BLPS^2\cdot \norm{T}+2\BLPS+2\norm{T}\cdot \mshoot\cdot |Q|)$ (recall that $\BLPS$ and $\mshoot$ are also polynomials in $|Q|,\norm{T}$).

    The ``if'' direction is fairly straightforward. Intuitively, since $x^\tar$ is large enough, then the number of repetitions of $\beta$ (namely $k$) is also large, and therefore the concatenation of $\theta$ actually yields a box-reaching path.
    
    Formally, assume the conditions of the lemma hold. Since $\alpha\beta^*\gamma\in \LPS$, then $|\alpha|,|\beta|,|\gamma|\le \BLPS$. Also recall that $|\theta|\le \norm{T}\cdot \mshoot\cdot |Q|$. Therefore, we have $\eff(\alpha\beta^k\gamma\theta)\le \norm{T}((k+2)\BLPS+\norm{T}\cdot \mshoot\cdot |Q|)$.
    We also have $\eff(\alpha\beta^k\gamma\theta)=x^\tar$ and $x^\tar\ge p_3(|Q|,\norm{T})\ge \norm{T}\cdot (\BLPS^2\cdot \norm{T}+2\BLPS+\norm{T}\cdot \mshoot\cdot |Q|)$ (this lower bound is essentially $p_3(|Q|,\norm{T})\ge p_3(|Q|,\norm{T})-\norm{T}\mshoot$, the actual bound described by $p_3$ is tighter in the second part of this proof), it follows that 
    \[
    \begin{split}
        &\norm{T}((k+2)\BLPS+\norm{T}\cdot \mshoot\cdot |Q|) \ge \norm{T}\cdot (\BLPS^2\cdot \norm{T}+2\BLPS+\norm{T}\cdot \mshoot\cdot |Q|)\implies \\
        &(k+2)\BLPS+\norm{T}\cdot \mshoot\cdot |Q| \ge \BLPS^2\cdot \norm{T}+2\BLPS+\norm{T}\cdot \mshoot\cdot |Q| \implies\\
        &(k+2)\BLPS\ge \BLPS^2\cdot \norm{T}+2\BLPS \implies k\ge \norm{T}\BLPS
    \end{split}
    \]
    And in particular, $k> \BLPS\norm{V}\ge \peak(\alpha)$ since $|\alpha|\le \BLPS$.

    Thus, since $\theta$ closes $\alpha\beta^*\gamma$ and $k$ is large enough, we have that $\alpha\beta^k\gamma\theta$ is a box-reaching path, and therefore $x^\tar$ is box-reachable.

    The converse direction is more involved. We start with a high-level sketch of the proof, depicted in \cref{fig:ell indices}.
    % \begin{figure}[ht]
    %     \centering
    %     \begin{tikzpicture}[xscale=0.5,yscale=0.3]
    % % % Draw grid
    % % \draw[step=1cm,gray,very thin] (0,0) grid (24,8);

    % % Draw path
    %  \draw[
    %      thick, blue,
    %      postaction={decorate}
    %  ]
    %  (0,4) -- (5,5) --  (2,3) -- (7,4) -- (4,2) -- (9,3) -- (6,1) -- (11,2) -- (8,0);

    % \draw[fill,blue] (0,4) circle [radius=0.2] node[above] {};
    
    % % Draw path
    %  \draw[
    %      thick,
    %      rounded corners=5pt,
    %      decoration={
    %          markings,
    %          mark=between positions 0 and 1 step 0.05 with {\arrow[scale=0.7]{>}}
    %      },
    %      postaction={decorate}
    %  ]
    %  (0,0) -- (22,0) -- (22,1) -- (12,1) -- (12,2) -- (20,2) -- (20,3) -- (17,3) -- (17,4) -- (24,4) -- (24,5) -- (16,5) -- (16,6) -- (24,6);

    % % Draw point marks
    % \foreach \x/\y/\label in {(0,0)/{},(24,6)/$\ell_1(x^\tar)$, (21,6)/$\ell_2$, (16,6)/$\ell_3$, (15,2)/$\ell_4$, (12,2)/$\ell_5$, (10,0)/$\ell_6$, (8,0)/$\ell_7$} {
    %     \draw[fill] \x circle [radius=0.2] node[above] {\label};
    % }
    % \end{tikzpicture}
    %     \caption{The sequence of indices defining box-safe suffixes. The $\yc$ axis is only for readability. The black path is a box-reaching path to $x^\tar$, the blue zig-zag is the LPS replacing the prefix up to $\ell_m$.}
    %     \label{fig:ell indices}
    % \end{figure}
    % %    
    Since $(x^\tar,q)$ is box reachable, there exists a box reaching path $(x_0,q_0)\Nleads{\pi}(x^\tar,q^\tar)$. Since $x^\tar$ is very large, then $\pi$ must be long, and in particular traverse many large $\xc$ values.
    
    Starting from $x^\tar$ and working backwards along $\pi$, we identify indices $\ell_m<\ell_{m-1}<\ldots <\ell_1$ such that the suffix $\pi[\ell_i+1,\ldots]$ remains entirely ``to the right'' of $\eff(\pi[1,\ell_i])$. That is, these suffixes are box reaching the paths.
    By working backwards enough, we can reach $\ell_m$ such that $\eff(\pi[1,\ell_m])$ is far enough from $x^\tar$ so that the overshoot of any linear path scheme in $\LPS$ is smaller than $x^\tar-\eff(\pi[1,\ell_m])$. 
    This guarantees that we can use a linear path scheme to reach the $\ell_m$ configuration, and continue from there with a box reaching suffix to $(x^\tar,q^\tar)$. Moreover, we can bound the length of this suffix to satisfy the conditions in the lemma.

    We proceed with the details.
    Let $\pi$ be a box-reaching path such that $(x_0,q_0)\Nleads{\pi}(x^\tar,q^\tar)$. 
    Since $x^\tar\ge p_3(|Q|,\norm{T})$, then $\pi$ visits large $\xc$ coordinates. 
    Fix some $m\in \bbN$, we define a sequence of indices $\ell_m<\ldots<\ell_1$ by setting $\ell_1= |\pi|$, and for every $1< j\le m$ define $\ell_j$ to be the maximal index where $\pi$ is to the left of its location at $\ell_{j-1}$. That is, $\ell_j=\max\{i\mid \eff(\pi[1,i])<\eff(\pi[1,\ell_{j-1}])\}$. 
    
    Observe that for every $1< j< m$ we have by definition that $\eff(\pi[\ell_j,\ell_{j-1}]) > 0$ and for every $i>\ell_{j}$ we have $\eff(\pi[1,i])\ge \eff(\pi[1,\ell_j])$ (otherwise $\ell_j$ would not be maximal).
   
Recall that $\norm{T}$ is the maximal entry in any transition of $\tup{Q,T}$. 
In particular, $\eff(\pi[\ell_j,\ell_{j-1}])\le \norm{T}$ (otherwise $\ell_j$ is not maximal) and thus $\eff(\pi[\ell_{m},\ell_{1}])\le m\norm{T}$. That is, the entire suffix of $\pi$ within this sequence is ``near'' $x^\tar$. However by the strict inequality of the $\ell_j$ we have $\eff(\pi[\ell_{m},\ell_{1}])\ge m$.

Take $m=\mshoot$. Since $x^\tar>p_3(|Q|,\norm{T})>\mshoot\cdot \norm{T}$ it follows that $\ell_m$ is well-defined (in that the $\max$ in the definition of $\ell_j$ is over a nonempty set -- we do not reach the ``beginning'' of $\pi$).

Note that $\eff(\pi[1,\ell_m])+\eff(\pi[\ell_m+1,\ldots])=x^\tar$, so by the above we have the following properties:
\begin{enumerate}
    \item $\eff(\pi[\ell_{m},\ell_{1}])\le \norm{T}\mshoot$.
    \item $\eff(\pi[\ell_{m},\ell_{1}])\ge \mshoot$.
    \item $\eff(\pi[1,\ell_m])\le x^\tar-\mshoot$.
    \item $\eff(\pi[1,\ell_m])\ge x^\tar-\norm{T}\mshoot$.
\end{enumerate}
Consider the configuration $(x',q')$ reached by $\pi[1,\ell_m]$ from $(0,q_0)$, then $x'\ge p_3(|Q|,\norm{T})-\norm{T}\mshoot$ by Property 4 above. 
By the exact same analysis as the ``if'' direction (here we used $p_3$ in a tight manner, since we ``lose'' $\norm{T}\mshoot$), any induced path $\alpha\beta^k\gamma$ of $\alpha\beta^*\gamma\in \LPS$ that reaches $(x',q')$ must satisfy $k\ge \peak(\alpha)$, and therefore $\oshoot(\alpha\beta^k\gamma)\le \mshoot$.
%\mitodo{Why therefore ...? $\oshoot(\alpha\beta^k\gamma)\le \mshoot$ only by definition of $\mshoot$}
%\shtodo{No, only for large enough $k$.}
%\mitodo{Right}
%\mitodo{Why do we know that such an LPS $\alpha\beta^k\gamma$ exists?}
%\shtodo{\cref{lem:blondin one VASS}}
Thus, by Property 3, we have
\[
\begin{split}
&\peak(\alpha\beta^k\gamma) =\oshoot(\alpha\beta^k\gamma)+\eff(\alpha\beta^k\gamma) \\
&\le\mshoot + \eff(\pi[1,\ell_m]) \le   \mshoot +x^\tar-\mshoot = x^\tar 
\end{split}
\]
Denote $\theta=\pi[\ell_{m},\ell_{1}]$, then $\theta$ is a box-safe suffix of $\pi$ (by the construction of the $\ell_i$). We can assume without loss of generality that the trace $(x',q')\leads{\theta}(x^\tar,q^\tar)$ does not repeat a configuration, otherwise an infix can be removed. All the configurations along this trace are in $[x^\tar- \eff(\theta),x^\tar]\times Q$ (since $\theta$ is box safe, and by the construction of the $\ell_j$), and by
%\mitodo{box-safety explains that x-values don't go beyond the target. So, it's worth adding that x-values don't go below $x^\tar- \eff(\theta)$, because of the definition of $l_i$.}
%\shtodo{added}
Property 1 we have $\eff(\theta)\le \norm{T}\cdot \mshoot$. By the pigeonhole principle, it follows that $|\theta|\le \norm{T}\cdot \mshoot\cdot |Q|$, satisfying the length constraints.

Finally, observe that $\theta$ closes $\alpha\beta^k\gamma$, since by Property 2 we have $\eff(\theta)\ge \mshoot\ge \oshoot(\alpha\beta^k\gamma)$. Thus, $(0,q_0)\Nleads{\alpha\beta^k\gamma\theta}(x^\tar,q^\tar)$.
\hfill\qed
%\end{proof}

\end{document}